\documentclass[journal]{IEEEtran}
\IEEEoverridecommandlockouts
\usepackage[noadjust]{cite}

\usepackage{amsmath,amssymb,amsfonts}
\usepackage{mathrsfs}
\usepackage{etoolbox}
\usepackage{algorithmic}
\usepackage{graphicx}
\usepackage{tikz}
\usepackage{pgfplots}
\pgfplotsset{compat=1.13}
\usepackage{float}
\usepackage{textcomp}
\usepackage{xcolor}
\usepackage[hidelinks]{hyperref}
\usepackage[all]{xy} 
\usepackage{bm}
\usepackage{multirow}
\usetikzlibrary{shapes,snakes,patterns,positioning,decorations.pathreplacing}

\pgfdeclarepatternformonly{soft horizontal lines}{\pgfpointorigin}{\pgfqpoint{100pt}{1pt}}{\pgfqpoint{100pt}{3pt}}%
{
	\pgfsetstrokeopacity{0.3}
	\pgfsetlinewidth{0.1pt}
	\pgfpathmoveto{\pgfqpoint{0pt}{0.5pt}}
	\pgfpathlineto{\pgfqpoint{100pt}{0.5pt}}
	\pgfusepath{stroke}
}

\pgfdeclarepatternformonly{soft crosshatch}{\pgfqpoint{-1pt}{-1pt}}{\pgfqpoint{4pt}{4pt}}{\pgfqpoint{3pt}{3pt}}%
{
	\pgfsetstrokeopacity{0.3}
	\pgfsetlinewidth{0.4pt}
	\pgfpathmoveto{\pgfqpoint{3.1pt}{0pt}}
	\pgfpathlineto{\pgfqpoint{0pt}{3.1pt}}
	\pgfpathmoveto{\pgfqpoint{0pt}{0pt}}
	\pgfpathlineto{\pgfqpoint{3.1pt}{3.1pt}}
	\pgfusepath{stroke}
}

\def\BibTeX{{\rm B\kern-.05em{\sc i\kern-.025em b}\kern-.08em
    T\kern-.1667em\lower.7ex\hbox{E}\kern-.125emX}}

\newcommand{\beq}{\begin{equation}}
\newcommand{\enq}{\end{equation}}
\newcommand{\bel}{\begin{lemma}}
\newcommand{\enl}{\end{lemma}}
\newcommand{\bet}{\begin{theorem}}
\newcommand{\ent}{\end{theorem}}

\newcommand{\suppress}[1]{}

\mathchardef\mhyphen="2D

\makeatletter
\newcommand*{\rom}[1]{\expandafter\@slowromancap\romannumeral #1@}
\makeatother

\newtheorem{remark}{Remark}
\newtheorem{theorem}{Theorem}
\newtheorem{lemma}{Lemma}
\newtheorem{corollary}{Corollary}
\newtheorem{proposition}{Proposition}

\makeatletter

\newcommand{\Rmnum}[1]{\expandafter\@slowromancap\romannumeral #1@}
\makeatother

\graphicspath{{./fig/}}
\DeclareGraphicsExtensions{.pdf}

\AtBeginEnvironment{equation}{\small}
\AtBeginEnvironment{equation*}{\small}
\AtBeginEnvironment{align}{\small}
\AtBeginEnvironment{align*}{\small}
\AtBeginEnvironment{gather}{\small}
\AtBeginEnvironment{gather*}{\small}
\AtBeginEnvironment{multline}{\small}
\AtBeginEnvironment{multline*}{\small}
\begin{document}
\title{Adaptive Coding for Two-Way Wiretap Channels with One-Sided Strong Secrecy}
\author{Masahito Hayashi and Yanling Chen 
\thanks{The work of MH was supported in part by the General R\&D Projects of 1+1+1 CUHK-CUHK(SZ)-GDST Joint Collaboration Fund (No. GRDP2025-022) and the Guangdong Provincial Quantum Science Strategic Initiative (No. GDZX2505003).} 
\thanks{This article was presented in part at the 2026 International Conference on Collaborative Technologies and Data Science for Smart City Applications (CODASSCA 2026), Yerevan, Armenia, September 7--9, 2026.}
\thanks{Masahito Hayashi is with the School of Data Science, The Chinese University of Hong Kong, Shenzhen, China, with the International Quantum Academy (SIQA), Shenzhen, China, and with the Graduate School of Mathematics, Nagoya University, Nagoya, Japan (e-mail: hmasahito@cuhk.edu.cn).}
\thanks{Yanling Chen is with Volkswagen Infotainment GmbH, Germany, 
Bochum, Germany
 (e-mail: yanling.chen@volkswagen-infotainment.com).}}

\maketitle

\begin{abstract}
We study a discrete memoryless two-way wiretap channel under strong
one-sided secrecy, where User~1's message is confidential and User~2's
message is not required to be secret.  We develop non-adaptive and
key-assisted adaptive inner bounds from a common resolvability estimate
for selected index components.  For the adaptive construction, we prove a
multiround leakage theorem that accounts both for information revealed about
a key before its later use and for the possibility that an encoder uses an
incorrectly decoded key.  The proof includes an ideal-to-implemented coupling,
deterministic codebook selection, and a growing-round argument that removes
the initialization-rate loss.  We eliminate the internal component rates
exactly and characterize the adaptive gain at each fixed auxiliary
distribution.  For an open two-parameter family of binary channels, we also
prove a converse for arbitrary non-adaptive stochastic codes and obtain an
operational adaptive-versus-non-adaptive separation.  For one representative
channel, the complete non-adaptive strong one-sided secrecy capacity region
is identified exactly.
\end{abstract}

\section{Introduction}
\label{sec:introduction}

Two-way communication allows each terminal to transmit and receive over the
same channel, and adaptation can create input dependence that is unavailable
to non-adaptive encoders \cite{Shannon1961,Dueck1979,Schalkwijk1983,Han1984,ZBS1986,PW1989, TU2007}.  When an
external eavesdropper is present, this interaction must be combined with
information-theoretic secrecy \cite{Shannon1949,Wyner1975}.  The normalized-leakage criterion commonly called weak secrecy allows the total leakage to remain nonzero or even to grow sublinearly with blocklength, whereas strong secrecy requires the total leakage to vanish \cite{src:Csiszar1996,MH2006}.  

In this paper, we consider a
discrete memoryless two-way wiretap channel under strong one-sided secrecy:
User~1's message must be protected, whereas User~2's message is not subject
to a secrecy constraint. The channel model is shown in Fig. \ref{fig:tw-wc-model}.  
This asymmetry is operationally significant because only User~1's payload
must satisfy the secrecy criterion, whereas User~2's unprotected payload can
also contribute randomness in the eavesdropper's marginal.
Here and throughout, \emph{payload} means the users' actual message
content, as distinguished from embedded keys, coding-randomization indices,
ciphertext coordinates, and dummy symbols used only for initialization or
code implementation.  A message variable remains operational message
information even when it is averaged in a resolvability argument.
User~2's payload is useful information decoded by its intended receiver, and
the same coordinate can also supply randomness in the distribution observed
by the eavesdropper.  It is not declared as public side information available to
the eavesdropper; rather, no secrecy guarantee is imposed on it.

Early work on two-way wiretap channels treated Gaussian and additive models
and developed non-adaptive achievable regions \cite{TY2007,TY2008}.  For the
general discrete memoryless model, adaptive key exchange was used to obtain
weak-secrecy regions \cite{GKYG2013}; one-sided and individual secrecy were
studied separately in \cite{QCHT2016,QDT2017}.  Strong secrecy has also been
analyzed through resolvability-based non-adaptive coding and key-assisted adaptive 
coding \cite{PB2011,HC2023,CH2025}.  These works provide the
coding foundations used here.  The present paper addresses two issues that
become decisive in the asymmetric setting: a complete treatment of leakage
propagation when a previously generated key is later used, and a comparison
of adaptive and non-adaptive inner bounds that remains valid after input
optimization and time sharing.

The multiround analysis has two coupled difficulties.  A ciphertext in
round $t$ uses a key generated in round $t-1$, so the eavesdropper's past observations
may already reveal information about that key.  Moreover, the implemented encoder
uses a decoded estimate of the key, whereas an ideal one-time-pad argument
uses the true key.  A proof must therefore track both the accumulated
confidential payload and the key retained for the next round, and then
transfer the ideal true-key estimate to the implemented decoded-key process.  At
the region level, a separate difficulty arises: strict enlargement at one
fixed auxiliary distribution need not imply strict enlargement of the
closed convexified union, because another non-adaptive distribution or time
sharing may cover the added boundary.

The main contributions are the following.
First, we isolate a resolvability estimate for selected index components and combine
it with complete-index decoding to obtain a reusable one-round mechanism.
The construction keeps operational payload, protected coordinates, and
coding-randomization coordinates distinct, even when an unprotected payload
is averaged in the eavesdropper's marginal.

Second, we develop a multiround propagation theorem for the key-assisted
adaptive construction.  Its state contains the accumulated confidential
payload and the key retained for the next round.  The theorem handles prior
key leakage, one-time-pad use with side information, and transfer from the
ideal true-key process to the implemented decoded-key process.  For a fixed
number of rounds the per-round estimates are exponential in the round
blocklength, while a separate growing-round argument establishes ordinary
achievability without asserting a positive exponent in total blocklength.

Third, we verify the abstract conditions for the concrete construction and
eliminate all component rates exactly.  At each strictly feasible auxiliary
distribution, the non-adaptive region is a rectangle and the adaptive region
is a containing polygon.  We characterize strict inclusion, the complete
User~1 gain profile, and the common maximum sum-rate.

Fourth, for an open family of binary channels with noiseless legitimate
cross-links, we prove a converse for the entire non-adaptive stochastic code
class defined in Section~\ref{sec:model-framework}.  The adaptive
construction attains a point outside the resulting operational capacity
region.  For one representative channel, an exact analytic maximization together
with the reverse inclusion induced by the uniform product input identifies
the non-adaptive capacity region exactly.

The exact rate elimination and fixed-distribution comparisons concern the
specific constructions developed here.  They are distinct from the binary
operational converse, which applies to arbitrary non-adaptive stochastic
codes in the stated code class.  User~2's payload is averaged in the
resolvability analysis but is not supplied to the eavesdropper as public side
information.  A supplementary appendix records the formal relation to an
earlier symmetric component-rate system; it is not part of the operational
main line.

The remainder of the paper is organized as follows.  Section~\ref{sec:model-framework} defines the
channel, code classes, secrecy criterion, achievable regions, and common
information quantities.  Section~\ref{sec:selective-resolvability} states the one-round resolvability
bound for selected index components.  Section~\ref{sec:nonadaptive} proves the
non-adaptive strong-secrecy inner bound.
Section~\ref{sec:one-round-selected-confidential} combines complete-index
decoding with that leakage estimate to formulate the general one-round
mechanism.  Section~\ref{sec:asymmetric-construction} instantiates the
mechanism as the direct asymmetric adaptive construction, and
Section~\ref{sec:embedded-key-propagation} treats the inter-round leakage
and decoded-key mismatch and verifies the required structural conditions.  Section~\ref{sec:region-geometry} proves the exact projection and
fixed-distribution geometry.  Section~\ref{sec:three-inner-bound-comparison} compares the weak, strong
non-adaptive, and adaptive construction-specific inner bounds, and
Section~\ref{sec:binary-family} proves an operational adaptive-versus-non-adaptive separation for a binary channel family.
Appendix~\ref{sec:symmetric-relation} records the formal relation to the earlier symmetric component-rate system.
Section~\ref{sec:conclusion} concludes the paper.  The appendices provide the one-round
resolvability proof, auxiliary multiround arguments, and constructive
projection details.

\section{Channel model, codes, and common information quantities}
\label{sec:model-framework}

\subsection{Channel model and codes}
We consider a discrete memoryless two-way wiretap channel with finite
input alphabets $\mathcal X_1,\mathcal X_2$, legitimate-output alphabets
$\mathcal Y_1,\mathcal Y_2$, and eavesdropper-output alphabet
$\mathcal Z$.  Its transition law is
$P_{Y_1Y_2Z\mid X_1X_2}$, and its operation is illustrated in
Fig.~\ref{fig:tw-wc-model}.  User~$i$ sends a message $M_i$ to the other
user and observes $Y_i$, while the external eavesdropper observes $Z$.
Only User~1's message is required to be secret.  We use $P_{Y_i\mid X_1X_2}$ and $P_{Z\mid X_1X_2}$ for the corresponding marginals of the joint transition law.

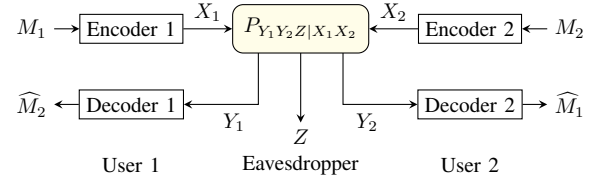
\begin{figure}[t]
\centering
\begin{tikzpicture}[>=stealth,scale=.80,transform shape]
  \node (m1) at (0,0) {$M_1$};
  \node [draw] (e1) at (1.65,0) {Encoder 1};
  \node [draw,rounded corners,fill=yellow!10,
    minimum width=2.25cm,minimum height=0.80cm] (ch) at (4.45,0)
    {$P_{Y_1Y_2Z\mid X_1X_2}$};
  \node [draw] (e2) at (7.25,0) {Encoder 2};
  \node (m2) at (8.9,0) {$M_2$};
  \node (mh2) at (0,-1.25) {$\widehat M_2$};
  \node [draw] (d1) at (1.65,-1.25) {Decoder 1};
  \node [draw] (d2) at (7.25,-1.25) {Decoder 2};
  \node (mh1) at (8.9,-1.25) {$\widehat M_1$};
  \path[->] (m1) edge (e1);
  \draw[->] (e1) -- (ch) node[midway,above] {$X_1$};
  \draw[->] (e2) -- (ch) node[midway,above] {$X_2$};
  \path[->] (m2) edge (e2);
  \path[->] (d1) edge (mh2);
  \path[->] (d2) edge (mh1);
  \draw[->] (3.75,-0.4) -- (3.75,-1.25) -- (d1);
  \draw[->] (5.15,-0.4) -- (5.15,-1.25) -- (d2);
  \draw[->] (4.45,-0.4) -- (4.45,-1.55);
  \node at (3.35,-1.55) {$Y_1$};
  \node at (5.55,-1.55) {$Y_2$};
  \node at (4.45,-1.78) {$Z$};
  \node at (1.65,-2.25) {User 1};
  \node at (7.25,-2.25) {User 2};
  \node at (4.45,-2.25) {Eavesdropper};
\end{tikzpicture}
\caption{Discrete memoryless two-way wiretap channel with an external
eavesdropper.}
\label{fig:tw-wc-model}
\end{figure}

For a block of $N$ channel uses, User~$i$ has a message set
$\mathcal M_{i,N}$, and $M_i$ is uniform on this set.  A rate sequence
$R_{i,N}$ means
\begin{equation}
 R_{i,N}:=\frac1N\log|\mathcal M_{i,N}|.
 \label{eq:model-rate-sequence}
\end{equation}
All logarithms are natural, so rates are measured in nats per channel
use.  Integer rounding of exponential alphabet sizes is suppressed.
We write $X_i^N=(X_{i,1},\ldots,X_{i,N})$ and similarly for the
legitimate and eavesdropper outputs.

A non-adaptive stochastic encoder for User~$i$ is a conditional law
$P_{X_i^N\mid M_i}$; its entire input sequence depends only on its
message and private local randomness.  The two local random seeds are
independent of each other and of the messages.  Shared encoder randomness
and correlated encoder seeds are not included in this code class, although
the deterministic public code design and a deterministic time-sharing
schedule may be common knowledge.  Thus, for every fixed public
non-adaptive code,
\begin{equation}
P_{X_1^NX_2^N\mid M_1M_2}
=P_{X_1^N\mid M_1}P_{X_2^N\mid M_2}.
\label{eq:model-nonadaptive-product-kernel}
\end{equation}
An adaptive stochastic encoder is a
sequence of causal kernels
\begin{equation}
 P_{X_{i,t}\mid M_i,Y_i^{t-1}},\qquad t=1,\ldots,N,
 \label{eq:model-adaptive-encoder}
\end{equation}
where private encoder randomness is included in these kernels.  Thus an
adaptive encoder may use its preceding channel outputs but not future
outputs.  A decoder at User~$i$ estimates the other user's message from
its own message, transmitted sequence, and received sequence; we write
\begin{equation}
 \widehat M_{3-i}=\psi_{i,N}(M_i,X_i^N,Y_i^N).
 \label{eq:model-decoder}
\end{equation}
The sequence $X_i^N$ is locally known to User~$i$, because it is the sequence actually generated and sent by that terminal; the decoder need not reconstruct the private random seed.
Allowing a stochastic decoder would not change the results, so the
decoders are taken to be deterministic.

An $N$-use code $\mathcal C_N$ consists of the two message sets, the two
encoders, and the two decoders.  It is called non-adaptive or adaptive
according to the encoders used.  Its average probability of decoding error is
\begin{equation}
 P_e(\mathcal C_N):=
 \Pr_{\mathcal C_N}\!\left[
   (\widehat M_1,\widehat M_2)\ne(M_1,M_2)
 \right].
 \label{eq:model-average-error}
\end{equation}
The probability averages over the uniform messages, private
randomness, and channel outputs.  For a deterministic selected
codebook, the codebook is part of $\mathcal C_N$ and is therefore fixed
in \eqref{eq:model-average-error}.

\subsection{Strong one-sided secrecy and achievability}
Only User~1's message is protected.  A sequence of deterministic codes
satisfies strong one-sided secrecy if
\begin{equation}
 I(M_1;Z^N\mid\mathcal C_N)\longrightarrow0.
 \label{eq:model-strong-one-sided}
\end{equation}
No secrecy constraint is imposed on $M_2$.  Thus secrecy of $M_2$ is not guaranteed, but $M_2$ is not thereby given to the eavesdropper as side information.  In particular, the criterion is not $I(M_1;Z^N,M_2\mid\mathcal C_N)\to0$.  Because a deterministic
code $\mathcal C_N$ is fixed, the conditioning in
\eqref{eq:model-strong-one-sided} merely records the public code and
may be suppressed after code selection.  Before selection, we instead
write $\mathsf C_N$ for a random codebook; quantities such as
$\mathbb E_{\mathsf C_N}I(M_1;Z^N\mid\mathsf C_N)$ are ensemble
averages and should not be confused with the leakage of one selected
code.

A rate pair $(R_1,R_2)$ is achievable by adaptive, respectively
non-adaptive, codes under strong one-sided secrecy if there is a
sequence $\{\mathcal C_N\}_{N\ge1}$ of the corresponding type such that
\begin{align}
 R_{i,N}&\longrightarrow R_i,\qquad i=1,2,\label{eq:model-rate-limit}\\
 P_e(\mathcal C_N)&\longrightarrow0,\label{eq:model-error-limit}\\
 I(M_1;Z^N\mid\mathcal C_N)&\longrightarrow0.
 \label{eq:model-leakage-limit}
\end{align}
The one-block constructions below initially use blocklength $n$.  The
adaptive construction uses $T$ such rounds and hence active blocklength
$nT$; when a code is embedded in an arbitrary total blocklength $N$, we
use fixed-input padding and calculate the rates with respect to $N$.
Thus every achievability statement is understood in the sense of
\eqref{eq:model-rate-limit}--\eqref{eq:model-leakage-limit}, not merely
along a restricted subsequence of total blocklengths.

The adaptive code constructed in Section~\ref{sec:asymmetric-construction} is a blockwise causal specialization of \eqref{eq:model-adaptive-encoder}.  Each round uses a fresh non-adaptive random codebook, while adaptation occurs across rounds through a key decoded in the preceding round.  Round~1 initializes the key and later rounds carry payload.

\subsection{Common information quantities}
Let $\mathcal Q$ be the set of distributions
\begin{equation}
P_{V_1}P_{V_2}P_{X_1\mid V_1}P_{X_2\mid V_2}.
\label{eq:model-Q}
\end{equation}
Operationally, $V_i^n$ is the auxiliary codeword selected from User~$i$'s random codebook, and $P_{X_i\mid V_i}^{\otimes n}$ is its memoryless stochastic prefix.  The product form describes independent codebook generation within one round; inter-round adaptation instead enters through the preceding decoded key in the current index. Note that the stochastic prefixes are absorbed into the induced channel, for instance,
\begin{equation}
	P_{Z\mid V_1V_2}=P_{Z\mid X_1X_2}\mathbin{\cdot}
	P_{X_1\mid V_1}\mathbin{\cdot}P_{X_2\mid V_2}.
	\label{eq:induced-eve-mac}
\end{equation}

Together with the channel, each $P\in\mathcal Q$ induces all information
quantities below.  For distributions on the same finite alphabet, define
\begin{equation}
D_{1+s}(P\Vert Q):=\frac1s\log\sum_zP(z)^{1+s}Q(z)^{-s},
\end{equation}
and, for a joint law $P_{ZX}$,
\begin{equation}
I_{1+s}^{\uparrow}(Z;X):=D_{1+s}(P_{ZX}\Vert P_ZP_X).
\label{eq:upward-renyi-definition}
\end{equation}
This is the convention of \cite[Eq.~(33)]{HC2023}.  

We next define the downward quantity used in the reliability bounds.  For
an input law $P_V$, a finite-output channel $W_{Y\mid V}$, and
$0<\alpha<1$, set
\begin{equation}
I_\alpha^{\downarrow}(Y;V)
:=\frac{\alpha}{\alpha-1}\log\sum_y
\left(\sum_v P_V(v)W_{Y\mid V}(y\mid v)^\alpha\right)^{1/\alpha}.
\label{eq:downward-renyi-definition}
\end{equation}
With $\alpha=\frac{1}{1+s}$, the Gallager random-coding function \cite{Gallager68} satisfies
$E_0(s,P_V,W)=sI_{\frac{1}{1+s}}^{\downarrow}(Y;V)$ under this convention.

Consider the case where $X$ is independent of $V$ and is
known to the decoder.  We define
\begin{equation}
I_\alpha^{\downarrow}(Y;V\mid X)
:=I_\alpha^{\downarrow}((X,Y);V),
\label{eq:conditional-downward-definition}
\end{equation}
where the channel from $V$ to $(X,Y)$ is
$P_X(x)W_{Y\mid VX}(y\mid v,x)$.  
Equivalently,
\begin{align}
&I_\alpha^{\downarrow}(Y;V\mid X)\notag\\
&\quad=\frac{\alpha}{\alpha-1}\log\sum_xP_X(x)\sum_y
 \left(\sum_vP_V(v)W_{Y\mid VX}(y\mid v,x)^\alpha
 \right)^{1/\alpha}.
\label{eq:conditional-downward-expanded}
\end{align}
Thus the logarithm is taken after averaging the Gallager sum over the
receiver side information, rather than separately for each value of $X$. 

For the two reliability decoders, define the induced single-user channels
\begin{align}
W_1(x_2,y_2\mid v_1)
&:=P_{X_2}(x_2)\sum_{x_1}P_{X_1\mid V_1}(x_1\mid v_1)
 P_{Y_2\mid X_1X_2}(y_2\mid x_1,x_2),\notag\\
W_2(x_1,y_1\mid v_2)
&:=P_{X_1}(x_1)\sum_{x_2}P_{X_2\mid V_2}(x_2\mid v_2)
 P_{Y_1\mid X_1X_2}(y_1\mid x_1,x_2),
\label{eq:induced-reliability-channels}
\end{align}
where $P_{X_i}$ is the marginal induced by
$P_{V_i}P_{X_i\mid V_i}$.  These channels include the
memoryless stochastic prefixes, and their outputs include the transmitted
sequence already known to the receiving user. 

For $s\in(0,1]$, define
\begin{align}
\begin{split}\label{eq:model-finite-s}
	A_s(P)&:=I_{1/(1+s)}^{\downarrow}(Y_2;V_1\mid X_2),\\
	B_s(P)&:=I_{1/(1+s)}^{\downarrow}(Y_1;V_2\mid X_1),\\
	C_s(P)&:=I_{1+s}^{\uparrow}(Z;V_1),\\
	D_s(P)&:=I_{1+s}^{\uparrow}(Z;V_2),\\
	E_s(P)&:=I_{1+s}^{\uparrow}(Z;V_1,V_2).
\end{split}
\end{align}
Their Shannon limits are
\begin{align}
\begin{split}\label{eq:model-Shannon}
	A(P)&:=I(Y_2;V_1\mid X_2),\\
	B(P)&:=I(Y_1;V_2\mid X_1),\\
	C(P)&:=I(Z;V_1),\\
	D(P)&:=I(Z;V_2),\\
	E(P)&:=I(Z;V_1,V_2).
\end{split}
\end{align}
Operationally, $A$ and $B$ are the two legitimate auxiliary-codeword rate limits, $C$ and $D$ are the individual resolvability costs at the eavesdropper, and $E$ is the joint resolvability cost.  Their finite-$s$ counterparts control the exponential bounds, while these Shannon limits govern the projected geometry.
For each fixed $P$, we have
\begin{equation}
(A_s,B_s,C_s,D_s,E_s)\longrightarrow(A,B,C,D,E)
\quad\text{as }s\downarrow0.
\label{eq:model-continuity}
\end{equation}
We use the strict feasibility set
\begin{equation}
\mathcal Q_{\rm F}:=
\left\{
	P\in\mathcal Q\left| 
		\begin{array}{l}
			C(P)<A(P),\\
			D(P)<B(P),\\ 
			E(P)<A(P)+B(P)
		\end{array}
		\right.
\right\}.
\label{eq:model-QF}
\end{equation}
These inequalities are precisely the nonemptiness conditions for the
strict auxiliary-rate systems used below.  A closed fixed-distribution
rate formula is invoked only for $P\in\mathcal Q_{\rm F}$; closure does
not create an achievable fixed-$P$ region when the corresponding strict
system is empty.

\subsection{Notation summary}
Table~\ref{tab:common-notation} collects the notation used across the
non-adaptive and adaptive parts.  The adaptive round variables are listed
separately in Table~\ref{tab:asym-round-operation}.
\begin{table}[t]
\centering
\caption{Common notation.}
\label{tab:common-notation}
\footnotesize
\begin{tabular}{p{1.45cm}p{5.45cm}}
\hline
Symbol & Meaning \\
\hline
$X_i,Y_i,Z$ & Legitimate user's inputs, outputs, and the eavesdropper's observation\\
$V_i$ & Auxiliary codeword variable before stochastic prefixing \\
$\mathcal Q,\mathcal Q_{\rm F}$ & Product auxiliary laws and strict feasibility subset \\
$A_s,B_s$ & Finite-$s$ reliability quantities \\
$C_s,D_s,E_s$ & Individual and joint finite-$s$ resolvability quantities \\
$A,B,C,D,E$ & Shannon limits of the preceding quantities \\
$\mathsf C_N,\mathsf C_t$ & Random block and round codebooks \\
$\mathcal R_{\rm N},\mathcal R_{\rm A}$ & Overall non-adaptive and adaptive inner bounds \\
\hline
\end{tabular}
\end{table}

\section{One-round resolvability for selected index components}
\label{sec:selective-resolvability}
This section isolates the resolvability estimate used by both the
non-adaptive and adaptive constructions.  The bound retains only specified components of the complete codebook indices
in the leakage criterion.  In one block, let User~$i$'s auxiliary
codeword be indexed by a triple $(A_i,J_i,B_i)$.  The component $A_i$ is
retained as protected information in the leakage criterion.  The component
$J_i$ carries information that is not protected by that criterion, whereas
$B_i$ is introduced solely for coding randomization.  Although $J_i$ and
$B_i$ have different operational meanings, both are averaged in
the eavesdropper's marginal.

The stochastic prefixes are included in the induced memoryless MAC
$P_{Z\mid V_1V_2}$ of \eqref{eq:induced-eve-mac}.  The three nonempty
transmitter subsets therefore correspond to the averaging rates
\begin{align}
\{1\}:&\quad R_{J_1}+R_{B_1}\quad\text{against }C_s(P),\notag\\
\{2\}:&\quad R_{J_2}+R_{B_2}\quad\text{against }D_s(P),\notag\\
\{1,2\}:&\quad R_{J_1}+R_{B_1}+R_{J_2}+R_{B_2}
\quad\text{against }E_s(P).
\label{eq:selected-subset-map}
\end{align}

\begin{proposition}[Resolvability for selected index components]
\label{prop:selected-index-component}
Fix $P\in\mathcal Q$ and $s\in(0,1]$.  Suppose that
$A_1,J_1,B_1,A_2,J_2,B_2$ are mutually independent and uniform before
codebook lookup, and that the two random codebooks are generated
independently according to $P_{V_1}^{\otimes n}$ and
$P_{V_2}^{\otimes n}$.  Then
\begin{align}
&\mathbb E_{\mathsf C}I(A_1,A_2;Z^n\mid\mathsf C)\notag\\
&\quad\leq\frac1s\Bigl[
 e^{ns(E_s(P)-R_{J_1}-R_{B_1}-R_{J_2}-R_{B_2})}\notag\\
&\qquad
 +e^{ns(C_s(P)-R_{J_1}-R_{B_1})}
 +e^{ns(D_s(P)-R_{J_2}-R_{B_2})}\Bigr].
\label{eq:selected-index-component-bound}
\end{align}
The same statement holds conditionally on a history variable whenever,
conditional on that history, the six coordinates retain the stated
independence and uniformity and the current codebook is independent of the
history with its original product ensemble law.
\end{proposition}
\begin{IEEEproof}
Appendix~\ref{app:selected-index-component-proof} gives the specialization of
the two-transmitter MAC resolvability argument in \cite{HC2023}.  The
symmetric notation above is obtained by treating $(J_i,B_i)$ as the
averaged super-index of transmitter $i$.  The conditional version follows
by applying the same bound for each history realization and then using the
tower property.
\end{IEEEproof}

The proposition is only a leakage estimate.  Section~\ref{sec:one-round-selected-confidential}
combines it with complete-index decoding to obtain the one-round
communication mechanism used later.

\section{Non-adaptive strong one-sided secrecy}
\label{sec:nonadaptive}
Fix $P\in\mathcal Q$.  In the non-adaptive construction, User~$i$
transmits an actual message $M_i$ of rate $R_i$.  Thus $(R_1,R_2)$ is the
payload-rate pair.  In addition, User~$i$ independently generates a
coding-randomization index of rate $R_{i,r}$.  This index carries no payload
and is introduced only to randomize the distribution observed by
the eavesdropper.  The complete codebook index of User~$i$ therefore has rate
$R_i+R_{i,r}$.  Each complete index selects an i.i.d. $V_i^n$ codeword,
followed by the memoryless stochastic prefix $P_{X_i\mid V_i}^{\otimes n}$.
The other user decodes the complete message and coding-randomization index
pair by maximum-likelihood decoding using its own transmitted sequence and
channel output.

\begin{proposition}[One-block non-adaptive bound]
\label{prop:nonadaptive-one-block}
For fixed $P\in\mathcal Q$, $s\in(0,1]$, and nonnegative rates, the random
codebook ensemble satisfies
\begin{align}
\mathbb E_{\mathsf C}P_e^n(\mathsf C)
&\leq a_n,
\label{eq:nonadaptive-error-ensemble}\\
\mathbb E_{\mathsf C}I(M_1;Z^n\mid\mathsf C)
&\leq b_n,
\label{eq:nonadaptive-leakage-ensemble}
\end{align}
where
\begin{align}
	a_n :=  
	& e^{ns(R_1+R_{1,r}-A_s(P))} +e^{ns(R_2+R_{2,r}-B_s(P))},
	\label{eq:nonadaptive-error-ensemble-an}\\
	b_n :=\frac1s\Bigl[
	&e^{ns(D_s(P)-R_{2,r}-R_2)}+e^{ns(C_s(P)-R_{1,r})}\notag\\
	&+e^{ns(E_s(P)-R_{1,r}-R_{2,r}-R_2)}\Bigr].\label{eq:nonadaptive-leakage-ensemble-bn}
\end{align}
Consequently, one deterministic codebook realization satisfies
\begin{equation}
P_e^n\leq2a_n,
\qquad I(M_1;Z^n)\leq2b_n.
\label{eq:nonadaptive-deterministic-consequence}
\end{equation}
\end{proposition}

\begin{IEEEproof}
The proof is given in Appendix \ref{app:nonadaptive-one-block-proof}.
\end{IEEEproof}

The following proposition is an algebraic elimination result for the
internal rates of this particular construction.  Its only-if direction is
not a converse for general non-adaptive codes.
\begin{proposition}[Exact elimination of the non-adaptive coding-randomization rates]
\label{prop:nonadaptive-strict-projection}
Fix $P\in\mathcal Q_{\rm F}$ and abbreviate
$A=A(P)$, $B=B(P)$, $C=C(P)$, $D=D(P)$, and $E=E(P)$.
For a payload-rate pair $(R_1,R_2)\in\mathbb R_+^2$, there exist
coding-randomization rates $(R_{1,r},R_{2,r})\in\mathbb R_+^2$ satisfying
\begin{align}
R_1+R_{1,r}&<A,&R_2+R_{2,r}&<B,
\label{eq:nonadaptive-shannon-reliability}\\
R_{1,r}&>C,&R_2+R_{2,r}&>D,
\label{eq:nonadaptive-shannon-individual}\\
R_{1,r}+R_2+R_{2,r}&>E
\label{eq:nonadaptive-shannon-joint}
\end{align}
if and only if
\begin{equation}
R_1<A-C,\qquad R_1<A+B-E,\qquad R_2<B.
\label{eq:nonadaptive-strict-projection}
\end{equation}
This equivalence characterizes only the projection of the internal-rate
system of the present construction; it does not bound arbitrary
non-adaptive codes.
\end{proposition}
\begin{IEEEproof}
Set
\begin{equation}
x:=R_{1,r},\qquad y:=R_2+R_{2,r}.
\label{eq:nonadaptive-xy-definition}
\end{equation}
The component-rate system is equivalent to
\begin{align}
C&<x<A-R_1,\label{eq:nonadaptive-x-interval}\\
D&<y<B,\label{eq:nonadaptive-y-interval}\\
R_2&\leq y,\label{eq:nonadaptive-y-payload}\\
x+y&>E.\label{eq:nonadaptive-xy-joint}
\end{align}
\textit{Only-if direction.}
The first interval gives $R_1<A-C$, while
\eqref{eq:nonadaptive-y-interval}--\eqref{eq:nonadaptive-y-payload}
give $R_2<B$.  Moreover, $x<A-R_1$, $y<B$, and $x+y>E$ imply
$R_1<A+B-E$.

\textit{If direction.}
Suppose \eqref{eq:nonadaptive-strict-projection} holds.  The first two
inequalities give
\begin{equation}
\max\{C,E-B\}<A-R_1.
\label{eq:nonadaptive-x-choice-condition}
\end{equation}
Choose
\begin{equation}
\max\{C,E-B\}<x<A-R_1.
\label{eq:nonadaptive-x-choice}
\end{equation}
Then $E-x<B$.  Since $P\in\mathcal Q_{\rm F}$ gives $D<B$ and
\eqref{eq:nonadaptive-strict-projection} gives $R_2<B$, choose
\begin{equation}
\max\{D,R_2,E-x\}<y<B.
\label{eq:nonadaptive-y-choice}
\end{equation}
With $R_{1,r}:=x$ and $R_{2,r}:=y-R_2$, all component-rate inequalities
hold strictly.  This proves the equivalence.
\end{IEEEproof}

\begin{remark}[Role of strict feasibility]
The restriction $P\in\mathcal Q_{\rm F}$ ensures that the strict
component-rate system is nonempty.  The closed fixed-distribution formula
below is used only for such $P$; closure after projection does not create an
achievable fixed-$P$ region when the underlying strict system is empty.
\end{remark}

For $P\in\mathcal Q_{\rm F}$ define
\begin{align}
\mathcal R_{\rm N}(P):=
\Bigl\{(R_1,R_2)&\in\mathbb R_+^2:R_2\leq B(P),\notag\\
&R_1\leq A(P)-C(P),\notag\\
&R_1\leq A(P)+B(P)-E(P)\Bigr\}.
\label{eq:nonadaptive-region}
\end{align}
This is the closure of the strict projection in
Proposition~\ref{prop:nonadaptive-strict-projection}.

\begin{theorem}[Non-adaptive achievable region]
\label{thm:nonadaptive-region}
The region
\begin{equation}
\mathcal R_{\rm N}:=\overline{\operatorname{conv}}\!\left(
\bigcup_{P\in\mathcal Q_{\rm F}}\mathcal R_{\rm N}(P)\right)
\label{eq:nonadaptive-overall}
\end{equation}
is achievable by non-adaptive codes under strong one-sided secrecy.
\end{theorem}
\begin{IEEEproof}
This is purely an achievability proof.  Proposition~\ref{prop:nonadaptive-strict-projection}
eliminates the internal coding-randomization rates of the proposed
construction; it does not provide a converse for arbitrary non-adaptive
codes.

\textit{Step 1: Fixed-$P$ interior achievability.}
Fix $P\in\mathcal Q_{\rm F}$ and an interior point $(R_1,R_2)$ of
$\mathcal R_{\rm N}(P)$.  Proposition~\ref{prop:nonadaptive-strict-projection}
provides nonnegative randomization rates for which
\eqref{eq:nonadaptive-shannon-reliability}--
\eqref{eq:nonadaptive-shannon-joint} hold with a common positive slack.
By \eqref{eq:model-continuity}, a sufficiently small fixed $s\in(0,1]$
preserves all five inequalities:
\begin{align}
R_1+R_{1,r}&<A_s(P),&R_2+R_{2,r}&<B_s(P),
\label{eq:nonadaptive-finite-s-reliability}\\
R_{1,r}&>C_s(P),&R_2+R_{2,r}&>D_s(P),
\label{eq:nonadaptive-finite-s-individual}\\
R_{1,r}+R_2+R_{2,r}&>E_s(P).
\label{eq:nonadaptive-finite-s-joint}
\end{align}
Every exponent in \eqref{eq:nonadaptive-error-ensemble-an} and
\eqref{eq:nonadaptive-leakage-ensemble-bn} is therefore negative.
Proposition~\ref{prop:nonadaptive-one-block} gives exponentially vanishing
ensemble-average error and leakage, and its deterministic consequence
selects one codebook satisfying both bounds.  Hence every fixed-$P$
interior point is achievable.

\textit{Step 2: Fixed-$P$ boundary and union.}
For a boundary point, choose nonnegative interior points
$(R_1^{(k)},R_2^{(k)})$ converging to it.  For each $k$, choose strict
component rates and a fixed $s_k>0$ as in Step~1.  Since the corresponding
exponential estimates hold for every sufficiently large blocklength, choose
a strictly increasing sequence $N_k$ such that for every $n\geq N_k$ the
$k$th construction has error and leakage at most $1/k$, with its integer
rounding error also at most $1/k$.  For $N_k\leq n<N_{k+1}$, use the $k$th
construction directly at blocklength $n$.  This gives codes at every
sufficiently large blocklength, with rates converging to the boundary point
and with vanishing error and leakage.  Thus every $\mathcal R_{\rm N}(P)$ is
achievable, and arbitrary choice of $P\in\mathcal Q_{\rm F}$ gives the union.

\textit{Step 3: Time sharing and outer closure.}
For finitely many achievable constituent points and weights, choose integer
segment lengths $n_\ell(N)$ summing to $N$ such that their normalized lengths
converge to the weights.  Omit zero-weight segments.  Use independent
messages, randomizers, selected codebooks, and channel blocks across
segments, under a deterministic public schedule.  The concatenated encoder
is non-adaptive and the total error is bounded by the sum of constituent
errors.  With $M^{[L]}=(M^{(1)},\ldots,M^{(L)})$ and
$Z^{[L]}=(Z^{(1)},\ldots,Z^{(L)})$, segment independence and the chain rule
give
\begin{equation}
I(M^{[L]};Z^{[L]}\mid C^{[L]})\leq
\sum_{\ell=1}^{L}I(M^{(\ell)};Z^{(\ell)}\mid C^{(\ell)}).
\label{eq:timesharing-leakage-sum}
\end{equation}
Thus every finite convex combination is achievable.

For a point in the outer closure, choose convex-hull points converging to it.
For the $k$th approximant choose a strictly increasing threshold $N_k$ such
that for all $N\geq N_k$ its time-sharing construction has rate error,
decoding error, and leakage at most $1/k$.  Use the $k$th construction when
$N_k\leq N<N_{k+1}$.  This all-blocklength sequence converges to the target
and proves \eqref{eq:nonadaptive-overall}.
\end{IEEEproof}

The non-adaptive construction retains protected information only in
User~1's index.  Key transport in the adaptive construction also requires
User~2's fresh key to be protected.  The next section formulates the
corresponding one-round mechanism before the rounds are connected through
an embedded key.

\section{One-round two-way wiretap mechanism}
\label{sec:one-round-selected-confidential}
The adaptive construction requires a one-round mechanism in which each
user transmits both information retained in the leakage criterion and index
components that are averaged in the eavesdropper's marginal.  We formulate that
mechanism here before introducing any inter-round key operation.

\subsection{Complete codebook indices and operational roles}
For User~$i$, write the complete codebook index as
\begin{equation}
(A_i,J_i,B_i).
\label{eq:one-round-index-triple}
\end{equation}
The component $A_i$ is retained as protected information in the one-round
leakage criterion.  The component $J_i$ carries information decoded by the
intended receiver but not protected by that criterion, whereas $B_i$ is
introduced solely for coding randomization.  The auxiliary codewords are selected as
\begin{equation}
V_1^n(A_1,J_1,B_1),\qquad V_2^n(A_2,J_2,B_2),
\label{eq:one-round-codeword-lookups}
\end{equation}
and are followed by the memoryless stochastic prefixes.  Each legitimate
receiver decodes the other user's complete triple.  In contrast, the
one-round secrecy criterion retains only $(A_1,A_2)$:
\begin{equation}
I(A_1,A_2;Z^n\mid\mathsf C).
\label{eq:one-round-selected-secrecy}
\end{equation}
Thus $J_i$ may be genuine message information even though it is averaged in
the eavesdropper's marginal.  The phrase ``not protected by the one-round criterion'' describes the role
of $J_i$ in the present analysis.  It
does not mean that its realized value is revealed to the eavesdropper as public
side information.  In particular, $J_i$ may be genuine payload decoded by
the intended receiver.  It is averaged in the eavesdropper's induced marginal
for the resolvability estimate, whereas $B_i$ is introduced only to
implement the code and carries no payload.

\subsection{Combined one-round bound}
\begin{proposition}[One-round bound with selected protected information]
\label{prop:one-round-selected-confidential}
Fix $P\in\mathcal Q$ and $s\in(0,1]$.  For mutually independent uniform
index coordinates and independent product random codebooks, define
\begin{align}
\epsilon_n(P,s):={}&
 e^{ns(R_{A_1}+R_{J_1}+R_{B_1}-A_s(P))}\notag\\
 &\ \ +e^{ns(R_{A_2}+R_{J_2}+R_{B_2}-B_s(P))},
\\
\delta_n(P,s):={}&\frac1s\Bigl[
 e^{ns(C_s(P)-R_{J_1}-R_{B_1})}
 +e^{ns(D_s(P)-R_{J_2}-R_{B_2})}\notag\\
&\qquad
 +e^{ns(E_s(P)-R_{J_1}-R_{B_1}-R_{J_2}-R_{B_2})}
 \Bigr].
\end{align}
Then the ensemble satisfies
\begin{align}
\mathbb E_{\mathsf C}P_e^n(\mathsf C)&\leq\epsilon_n(P,s),\label{eq:one-round-error-bound}\\
\mathbb E_{\mathsf C}I(A_1,A_2;Z^n\mid\mathsf C)&\leq\delta_n(P,s).\label{eq:one-round-leakage-bound}
\end{align}
In particular, both quantities decay exponentially when
\begin{align}
R_{A_1}+R_{J_1}+R_{B_1}&<A_s(P),\notag\\
R_{A_2}+R_{J_2}+R_{B_2}&<B_s(P),\notag\\
R_{J_1}+R_{B_1}&>C_s(P),\notag\\
R_{J_2}+R_{B_2}&>D_s(P),\notag\\
R_{J_1}+R_{B_1}+R_{J_2}+R_{B_2}&>E_s(P).
\label{eq:one-round-five-conditions}
\end{align}
\end{proposition}
\begin{IEEEproof}
Complete-index maximum-likelihood decoding through the two induced
reliability channels gives \eqref{eq:one-round-error-bound}.  Proposition~\ref{prop:selected-index-component}
gives \eqref{eq:one-round-leakage-bound}.
\end{IEEEproof}

\begin{corollary}[History-conditioned one-round bound]
\label{cor:history-conditioned-one-round}
Let $\mathscr H$ be a history variable that excludes the current random
codebook.  If, conditional on every admissible history realization, the six
current index coordinates have the product uniform law of Proposition~\ref{prop:one-round-selected-confidential}
and the current codebook remains independent with its original ensemble
law, then
\begin{equation}
\mathbb E_{\mathsf C_t}I(A_1,A_2;Z_t\mid\mathscr H=h,\mathsf C_t)
\leq\delta_n(P,s)
\label{eq:history-conditioned-one-round}
\end{equation}
for every admissible $h$.
\end{corollary}

The non-adaptive construction is the specialization in which User~2's
protected index component is a singleton.  In the adaptive construction that is to be introduced in the next section,
User~2's protected index component is instead a fresh key for the next
round.  

\section{Asymmetric key-assisted adaptive construction}
\label{sec:asymmetric-construction}
This section constructs the adaptive code by applying the one-round
mechanism of Section~\ref{sec:one-round-selected-confidential} in every
round.  The section specifies the operational pipeline and the substitution
of the round variables into the general three-component codebook indices.  The
round-to-round leakage issue is deferred to
Section~\ref{sec:embedded-key-propagation}.

\subsection{Payload rounds}
Round~1 is used only to initialize the inter-round key and carries no actual
message content.  Rounds $t\geq2$ are payload rounds.  In each such round,
both users transmit fresh actual messages, while User~2 also transports the
key to be used in the following round.

Fix $P\in\mathcal Q$, a number of rounds $T\geq2$, and a per-round
blocklength $n$.  Independent random codebooks are generated for the rounds.
For every payload round $t\geq2$, the message variables in round $t$ are
\begin{equation}
	(S_t,U_t,M_2^t),
	\label{eq:asym-payload-variable-list}
\end{equation}
where $(S_t,U_t)$ is User~1's confidential payload, i.e.,
\begin{equation}
	M_1^t=(S_t,U_t).
	\label{eq:asym-current-message}
\end{equation}
The variable $M_2^t$ is User~2's unprotected payload. In particular, the component
$S_t$ is protected directly by the current-round random-coding and
resolvability mechanism, whereas $U_t$ is protected through a one-time pad
using the key generated in the preceding round.  Moreover, User~2's variable $M_2^t$ is genuine message content that is to be decoded by User~1,
but no secrecy guarantee is imposed on it.  In particular, its use as an
averaging coordinate in the resolvability bound does not turn it into public
side information.

In round $t,$ User~2 also generates a fresh key $K_t$ for the next round. Note that $K_t$ is an inter-round coding resource rather than payload.
The variables $U_t$ and $K_t$ take values in the same finite abelian group $G_n$. Besides the message payloads and key variables, $O_{1,t},O_{2,t}$ are coding-randomization coordinates generated at User~1 and User~2, respectively.

The implemented and ideal ciphertexts are
\begin{align}
C_t^{\rm real}&=U_t\oplus\widehat K_{t-1},
\label{eq:asym-real-ciphertext}\\
C_t&=U_t\oplus K_{t-1}.
\label{eq:asym-ideal-ciphertext}
\end{align}
The ciphertext
$C_t^{\rm real}$ is a derived transmitted coordinate, not an additional
message.  Thus only $(S_t,U_t,M_2^t)$ contributes to the payload rates
$(R_1,R_2)$.

In the implemented round, the one-round index triples are
\begin{align}
(A_1,J_1,B_1)&=(S_t,C_t^{\rm real},O_{1,t}),\notag\\
(A_2,J_2,B_2)&=(K_t,M_2^t,O_{2,t}).
\label{eq:asym-round-indices}
\end{align}
For the ideal leakage analysis, $C_t^{\rm real}$ is replaced by $C_t$.
Consequently, the protected components retained by the current one-round leakage criterion
are $(S_t,K_t)$, whereas the ciphertext $C_t$ and User~2's payload
$M_2^t$ are unprotected information components.  The variables
$O_{1,t},O_{2,t}$ are coding randomization.

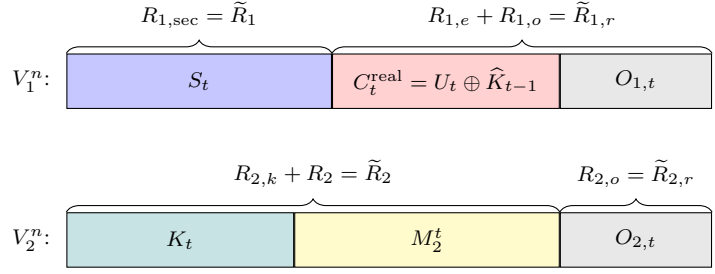
\begin{figure*}[t]
\centering
\begin{tikzpicture}[font=\footnotesize]
  \node[anchor=east] at (-0.25,1.05) {$V_1^n$:};
  \node[draw,fill=blue!22,minimum width=3.5cm,minimum height=0.72cm] (s) at (1.6,1.05) {$S_t$};
  \node[draw,fill=red!18,minimum width=3cm,minimum height=0.72cm,right=0pt of s] (c) {$C_t^{\rm real}=U_t\oplus\widehat K_{t-1}$};
  \node[draw,fill=gray!18,minimum width=2.0cm,minimum height=0.72cm,right=0pt of c] (o1) {$O_{1,t}$};
  \draw[decorate,decoration={brace,amplitude=5pt}] (s.north west) -- (s.north east) node[midway,above=6pt] {$R_{1,\mathrm{sec}}=\widetilde R_1$};
  \draw[decorate,decoration={brace,amplitude=5pt}] (c.north west) -- (o1.north east) node[midway,above=6pt] {$R_{1,e}+R_{1,o}=\widetilde R_{1,r}$};
  \node[anchor=east] at (-0.25,-1.05) {$V_2^n$:};
  \node[draw,fill=teal!22,minimum width=3cm,minimum height=0.72cm] (k) at (1.35,-1.05) {$K_t$};
  \node[draw,fill=yellow!25,minimum width=3.5cm,minimum height=0.72cm,right=0pt of k] (m) {$M_2^t$};
  \node[draw,fill=gray!18,minimum width=2.0cm,minimum height=0.72cm,right=0pt of m] (o2) {$O_{2,t}$};
  \draw[decorate,decoration={brace,amplitude=5pt}] (k.north west) -- (m.north east) node[midway,above=6pt] {$R_{2,k}+R_2=\widetilde R_2$};
  \draw[decorate,decoration={brace,amplitude=5pt}] (o2.north west) -- (o2.north east) node[midway,above=6pt] {$R_{2,o}=\widetilde R_{2,r}$};
\end{tikzpicture}
\caption{Complete codebook-index structure in a payload round $t\geq2$.
For User~1, $S_t$ is directly protected payload, $C_t^{\rm real}$ carries
the one-time-pad-protected payload $U_t$, and $O_{1,t}$ is coding
randomization.  For User~2, $K_t$ is the protected next-round key, $M_2^t$
is unprotected payload, and $O_{2,t}$ is coding randomization.  The braces
show the internal rates seen by the one-round decoder.}
\label{fig:asymmetric-round-structure}
\end{figure*}
\begin{table*}[t]
\centering
\caption{Operational roles in payload round $t\geq2$.}
\label{tab:asym-round-operation}
\footnotesize
\setlength{\tabcolsep}{4pt}
\begin{tabular}{p{5cm}p{5 cm}p{5 cm}}
\hline
Role & User~1 & User~2 \\
\hline
Protected index component & $S_t$ & Fresh key $K_t$ for round $t+1$ \\
Unprotected information component & $C_t^{\rm real}=U_t\oplus\widehat K_{t-1}$ & Payload $M_2^t$ \\
Coding-randomization component & $O_{1,t}$ & $O_{2,t}$ \\
Complete decoded tuple & User~1 decodes $(K_t,M_2^t,O_{2,t})$ & User~2 decodes $(S_t,C_t^{\rm real},O_{1,t})$ \\
Post-processing & Retains $\widehat K_t$ & Uses the true $K_{t-1}$ to recover $U_t$ \\
\hline
\end{tabular}
\end{table*}

\subsection{Rates and key matching}
Assign the rates
\begin{align}
R_{A_1}&=R_{1,\mathrm{sec}},&R_{J_1}&=R_{1,e},&R_{B_1}&=R_{1,o},\notag\\
R_{A_2}&=R_{2,k},&R_{J_2}&=R_2,&R_{B_2}&=R_{2,o}.
\label{eq:asym-rate-identification}
\end{align}
The payload rate at User~1 is
\begin{equation}
R_1=R_{1,\mathrm{sec}}+R_{1,e}.
\label{eq:asym-R1}
\end{equation}
Only the key consumed in the following one-time pad is transported.  Hence
we choose the common group so that
\begin{equation}
R_{2,k}=R_{1,e}.
\label{eq:asym-key-size}
\end{equation}
This normalization loses no payload-rate point: if a feasible component
assignment has $R_{2,k}>R_{1,e}$, reducing $R_{2,k}$ to $R_{1,e}$ decreases
only User~2's decoding load and leaves all averaging-rate inequalities
unchanged.

For compatibility with the projection notation, define
\begin{align}
\widetilde R_1&=R_{1,\mathrm{sec}},&
\widetilde R_{1,r}&=R_{1,e}+R_{1,o},\notag\\
\widetilde R_2&=R_2+R_{2,k},&
\widetilde R_{2,r}&=R_{2,o}.
\label{eq:asym-underlying-rates}
\end{align}
The rates $R_1$ and $R_2$ measure actual message content.  The component
rates $R_{1,\mathrm{sec}}$, $R_{1,e}$, $R_{2,k}$, $R_{1,o}$, and $R_{2,o}$
describe how payload and coding resources are allocated inside a round.  By
contrast, $\widetilde R_i$ and $\widetilde R_{i,r}$ are the main- and
randomization-index rates seen by the underlying one-block code.  They are
internal index rates and should not be interpreted as additional payload
rates.
Substitution into \eqref{eq:one-round-five-conditions} gives
\begin{align}
\begin{split}\label{eq:asym-finite-s-conditions}
\widetilde R_1+\widetilde R_{1,r}&<A_s(P),\\
\widetilde R_2+\widetilde R_{2,r}&<B_s(P),\\
\widetilde R_{1,r}&>C_s(P),\\
\widetilde R_{2,r}+R_2&>D_s(P),\\
\widetilde R_{1,r}+\widetilde R_{2,r}+R_2&>E_s(P).
\end{split}
\end{align}
The corresponding per-round leakage increment is
\begin{align}
\delta_n(P,s)=\frac1s\Bigl[&
 e^{ns(C_s(P)-\widetilde R_{1,r})}
 +e^{ns(D_s(P)-R_2-\widetilde R_{2,r})}\notag\\
&+e^{ns(E_s(P)-\widetilde R_{1,r}-R_2-\widetilde R_{2,r})}\Bigr],
\label{eq:asym-delta}
\end{align}
and the union-bound reliability estimate over $T$ rounds is
\begin{equation}
a_{n,T}:=T\left[
 e^{ns(\widetilde R_1+\widetilde R_{1,r}-A_s(P))}
 +e^{ns(\widetilde R_2+\widetilde R_{2,r}-B_s(P))}\right].
\label{eq:asym-ensemble-error}
\end{equation}
These expressions are direct instances of Proposition~\ref{prop:one-round-selected-confidential}; they are not separately rederived here.

\subsection{Initialization}
Round~1 carries no actual payload.  The same one-round mechanism is used
under the explicit specialization
\begin{align}
(A_1,J_1,B_1)
=&(D_{1,\mathrm{sec}},D_{1,e},O_{1,1}),\notag\\
(A_2,J_2,B_2)
=&(K_1,D_2,O_{2,1}),
\label{eq:asym-initialization-map}
\end{align}
where the $D$-coordinates are independent uniform dummies with the
corresponding payload-round sizes.  $K_1$ is the protected User~2 index component retained in the
initialization leakage criterion.  User~1 decodes and retains
$\widehat K_1$, while User~2 retains the true key it generated.  Discarding
the dummy protected component by data processing leaves the initial-key
leakage bound required for the recursion.  Thus all rounds have the same
codebook-index dimensions.  For $T$ rounds, rounds $2$ through $T$ carry
payload, giving the effective factor $(T-1)/T$.

\section{Leakage propagation with imperfect embedded keys}
\label{sec:embedded-key-propagation}
The one-round theorem of
Section~\ref{sec:one-round-selected-confidential} already controls the
fresh protected information in each ideal round.  The remaining issue is
whether these local bounds propagate when the eavesdropper's past observations
may contain information about the preceding key and when the implemented
encoder may use an incorrectly decoded key.

Denote the secret state after round $t$ as $(W_t,K_t)$, containing all
of User~1's actual messages transmitted so far and the key retained for the
following round.  Then one round update of the secret state can be described as
\begin{equation}
(W_{t-1},K_{t-1})\longrightarrow(W_t,K_t).
\label{eq:abstract-state-transition}
\end{equation}
The target recursion is
\begin{align}
&I(W_t,K_t;Z^{1:t}\mid\boldsymbol{\mathsf C})\leq
I(W_{t-1},K_{t-1};Z^{1:t-1}\mid\boldsymbol{\mathsf C})+\delta_{n,t}.
\label{eq:abstract-target-recursion}
\end{align}
The first term is the leakage accumulated before round $t$, and
$\delta_{n,t}$ is the one-round leakage of the fresh protected state
$Q_t=(S_t,K_t)$ in round $t$.  There is no separate increment for $U_t$: its current
leakage is absorbed into the information already accumulated about
$K_{t-1}$.

The proof therefore has two genuinely multiround stages.  First, the ideal
process with $C_t=U_t\oplus K_{t-1}$ is analyzed by a conditional one-time-pad
argument.  Second, the ideal estimate is transferred to the implemented
process with $C_t^{\rm real}=U_t\oplus\widehat K_{t-1}$.  The two coupled
processes agree unless a wrongly decoded key is used.  The structural
round-kernel conditions and the local one-round theorem play distinct
roles; neither is a consequence of channel memorylessness alone.

Moreover, for propagating a one-round leakage estimate through an ideal embedded-key process, certain conditions have to be fulfilled.  
We first prove a general propagation theorem under these conditions.  
We then verify that the ideal process and random round-codebook ensemble induced by the 
asymmetric adaptive code of Section~\ref{sec:asymmetric-construction}
satisfy all these conditions.  Consequently, the fixed-round and growing-round
reliability and secrecy bounds obtained here apply to the concrete asymmetric adaptive code. 

The proof dependency is as follows.  Proposition~\ref{prop:selected-index-component}
and Corollary~\ref{cor:history-conditioned-one-round} supply the fresh
one-round leakage increment.  Proposition~\ref{prop:abstract-otp} absorbs
the encrypted-payload leakage into the leakage already carried by the
preceding key, and Lemma~\ref{lem:one-step-recursion} gives the ideal
one-step recursion.  Conditions (A1a)--(A4b) are then verified for the
concrete asymmetric construction.  Finally, the real/ideal coupling and
continuity bound transfer the recursion to the implemented decoded-key
process; the fixed-round, growing-round, and rate-projection steps then
complete the achievability proof.

\subsection{Variables used in the multiround leakage proof}

Fix a total number of rounds $T\geq2$, and label the rounds by
$t\in\{1,\ldots,T\}$.  
Let
\begin{equation}
	\boldsymbol{\mathsf C}:=(\mathsf C_1,\ldots,\mathsf C_T)
\end{equation}
denote the collection of independently generated round codebooks, where
$\mathsf C_t$ is the codebook used in round $t$.

For each round $t$, let
$Z_t\in\mathcal Z^n$ denote the length-$n$ observation block at the eavesdropper in that
round, and write $Z^{1:t}:=(Z_1,\ldots,Z_t)$ for the tuple of round
outputs.  

For each payload round $t\in\{2,\ldots,T\}$, define the fresh protected
state by
\begin{equation}
Q_t:=(S_t,K_t),
\end{equation}
and User~1's messages sent through the current round by
\begin{equation}
	W_t:=(S_2,U_2,\ldots,S_t,U_t).
\end{equation}

Round~1 as the initialization round is treated separately: it carries no actual message (thus $W_1=\emptyset$), and only the
initial key $K_1$ is retained in the state used in the leakage recursion.
In particular,
$(W_t,K_t)$ and $(W_{t-1},Q_t,U_t)$ contain the same variables up
to ordering, with no duplicated terminal key. 

The terminal key $K_T$ is retained only to keep the recursive state uniform
through round $T$ and is discarded at the end of the leakage analysis.

For each round, define the history excluding the current random codebook by
\begin{equation}
\mathscr H_{t-1}^{-t}:=
\sigma\!\left(W_{t-1},K_{t-1},Z^{1:t-1},
\boldsymbol{\mathsf C}_{<t},\boldsymbol{\mathsf C}_{>t}\right).
\label{eq:abstract-history-minus-current}
\end{equation}
After a realization of $\mathsf C_t$ is fixed, conditioning on
$(\mathscr H_{t-1}^{-t},\mathsf C_t)$ is equivalent to conditioning on the
preceding state, observations, and the complete codebook collection.  The
exclusion of $\mathsf C_t$ is essential when the one-round ensemble bound is
applied.

\subsection{Probability levels and conditions for leakage propagation}
\label{subsec:probability-levels}
The conditioning order is part of the proof.  We distinguish three
probability levels throughout the propagation argument.

\textit{Level 1: fixed-codebook process.}  In the \emph{fixed-codebook process}, the complete collection is
fixed at $\boldsymbol{\mathsf C}=\boldsymbol c$.  The source variables,
stochastic-prefix randomness, and channel randomness then generate the
ideal or implemented process.  Conditions (A1a), (A1b), and (A3), as well
as the one-time-pad and one-step recursion lemmas, are structural statements
about the induced finite-alphabet conditional laws at this level.  All
history-conditioned assertions are restricted to positive-probability
histories.

\textit{Level 2: history-conditioned current-round ensemble.}  Second, in the \emph{history-conditioned current-round ensemble}, an
admissible value $h$ of $\mathscr H_{t-1}^{-t}$ is fixed while only the
current codebook $\mathsf C_t$ remains random.  By condition (A2), its
conditional law is its original random-coding ensemble law.  The
history-conditioned resolvability estimate and condition (A4a) are applied
at this level.

\textit{Level 3: full ensemble.}  Third, the history and the codebooks outside round $t$ are averaged by the
tower property, producing the full-ensemble recursion.  Only after the
ideal recursion and the transfer from the ideal process to the implemented process are complete do we select
one deterministic codebook collection.  We write $\mathsf C_t$ and
$\boldsymbol{\mathsf C}$ for random codebooks, $\mathsf c_t$ and
$\boldsymbol c$ for fixed realizations, and display the expectation
subscript whenever a codebook average is taken.

We next state conditions (A1a), (A1b), and (A2)--(A4b), summarized in
Table~\ref{tab:condition-use-verification-map}, under which the ideal
embedded-key process admits the one-step leakage recursion. 
Note that Conditions (A1a)--(A3) are
structural properties that need to be fulfilled by the code construction, whereas (A4a)--(A4b) provides the
quantitative one-round leakage bounds.

\medskip
\noindent\textbf{Condition (A1a): independent generation of the current-round variables.}
For every payload round $t\in\{2,\ldots,T\}$, the joint distribution of
$(S_t,U_t,K_t)$, the preceding state, and the past observation
has the following properties.  The variables
$(S_t,U_t,K_t)$ are mutually independent, $U_t$ and
$K_t$ are uniform on their respective finite groups, and
\begin{equation}
(S_t,U_t,K_t)
\perp (W_{t-1},K_{t-1},Z^{1:t-1})
\mid\boldsymbol{\mathsf C}.
\label{eq:abstract-current-freshness}
\end{equation}
In particular, with $A_t=(W_{t-1},K_{t-1})$, this gives
$Q_t\perp(A_t,Z^{1:t-1})\mid\boldsymbol{\mathsf C}$, which is the
independence property used to bound the one-round leakage term for $Q_t$.

\medskip
\noindent\textbf{Condition (A1b): independence of $K_{t-1}$ and $W_{t-1}$ before conditioning on past observations.}
For every payload round $t\in\{2,\ldots,T\}$, the joint distribution of
$(W_{t-1},K_{t-1})$, before $Z^{1:t-1}$ is adjoined to the
conditioning, satisfies
\begin{equation}
K_{t-1}\perp W_{t-1}
\mid\boldsymbol{\mathsf C}.
\label{eq:abstract-preobservation-closure}
\end{equation}
This is not a posterior-independence assumption.  After conditioning on
$Z^{1:t-1}$, the eavesdropper's observation may correlate with the key and the accumulated
message; that correlation is retained in the recursion.

\medskip
\noindent\textbf{Condition (A2): independence of the current-round codebook.}
For every payload round $t\in\{2,\ldots,T\}$, the random round-codebook
ensemble satisfies the following requirement: the current codebook
$\mathsf C_t$ is independent of $\mathscr H_{t-1}^{-t}$ and, conditional on
each history realization, retains its original random-coding ensemble law.

\medskip
\noindent\textbf{Condition (A3): dependence of the current output on the current-round variables.}
Here $C_t$ denotes the ideal ciphertext, $\mathsf C_t$ the current random
codebook, and $\boldsymbol{\mathsf C}$ the full codebook collection.  The
first kernel below conditions on a fixed ciphertext; the second is obtained
only after averaging the history-wise uniform ciphertext.
For every payload round $t\in\{2,\ldots,T\}$, fix an admissible
history realization $h$ and a current-codebook realization $\mathsf c_t$.
After the fresh coding-randomization indices have been averaged, the
conditional law generating $Z_t$ in the ideal fixed-codebook process
satisfies the following identities for every positive-probability history
$h$ and every $\mathsf c_t$ in the support of the current codebook ensemble:
\begin{align}
&P_{Z_t\mid W_{t-1},Z^{1:t-1},Q_t,C_t,
 \boldsymbol{\mathsf C}}
 =P_{Z_t\mid Q_t,C_t,\mathsf C_t},
\label{eq:abstract-round-kernel}\\
&P_{Z_t\mid W_{t-1},Z^{1:t-1},Q_t,
 \boldsymbol{\mathsf C}}
 =P_{Z_t\mid Q_t,\mathsf C_t}.
\label{eq:abstract-averaged-kernel}
\end{align}
The first equality \eqref{eq:abstract-round-kernel} states that, once $(Q_t,C_t)$ is fixed,
the current output has no further direct dependence on the past.  The second equality \eqref{eq:abstract-averaged-kernel} 
states that, after averaging the history-wise uniform
ciphertext, the current output has no dependence on the preceding message variables.
In addition,
\begin{equation}
U_t-(C_t,Z^{1:t-1},W_{t-1},Q_t,
\boldsymbol{\mathsf C})-Z_t
\label{eq:abstract-otp-markov}
\end{equation}
is a Markov chain.  Thus the current channel does not access $U_t$ directly;
it sees that message part only through $C_t$.  These three statements are
structural properties of the round kernel and are logically independent of
the leakage estimate below.

\medskip
\noindent\textbf{Condition (A4a): one-round leakage bound.}
For every payload round $t\in\{2,\ldots,T\}$, the current-codebook ensemble
satisfies the following bound uniformly over the preceding history.  There
is a nonnegative number $\delta_{n,t}$ such that, for every
positive-probability realization $h$ of $\mathscr H_{t-1}^{-t}$,
\begin{equation}
\mathbb E_{\mathsf C_t}
 I(Q_t;Z_t\mid \mathscr H_{t-1}^{-t}=h,\mathsf C_t)
\leq\delta_{n,t}.
\label{eq:abstract-one-round-bound}
\end{equation}
The expectation is only over the current codebook.  By (A2), the
conditional law of $\mathsf C_t$ given
$\mathscr H_{t-1}^{-t}=h$ is its original ensemble law, and the right-hand
side is uniform in $h$.  Averaging the fixed-history inequality over the history and the codebooks
outside round $t$, and then restoring $\mathsf C_t$ to the conditioning,
gives by the tower property
\begin{equation}
\mathbb E_{\boldsymbol{\mathsf C}}
 I(Q_t;Z_t\mid W_{t-1},K_{t-1},Z^{1:t-1},
 \boldsymbol{\mathsf C})
\leq\delta_{n,t}.
\label{eq:abstract-one-round-tower}
\end{equation}
This is the increment appearing in the ensemble recursion.

\medskip
\noindent\textbf{Condition (A4b): initial-key leakage bound.}
The initialization round and its round-1 codebook ensemble supply the
initial key state and satisfy
\begin{equation}
\mathbb E_{\boldsymbol{\mathsf C}}
 I(K_1;Z_1\mid\boldsymbol{\mathsf C})
\leq\delta_{n,1}.
\label{eq:abstract-initial-bound}
\end{equation}
Conditions (A4a)--(A4b) are the only points at which the one-round coding
result enters the abstract propagation argument.  For the concrete code,
(A4a) follows from Corollary~\ref{cor:history-conditioned-one-round}, and
(A4b) follows from its initialization specialization
\eqref{eq:asym-initialization-map}.

\begin{table*}[t]
\centering
\caption{Conditions used in the leakage-propagation argument and their
verification for the asymmetric construction.}
\label{tab:condition-use-verification-map}
\footnotesize
\setlength{\tabcolsep}{5pt}
\renewcommand{\arraystretch}{1.12}
\begin{tabular}{p{1.5cm}p{8cm}p{7.5cm}}
\hline
Condition & Role in the propagation proof & Concrete fact establishing it \\
\hline
(A1a)--(A1b) & Use the stated independence relations and apply the conditional one-time pad without
assuming posterior independence after $Z^{1:t-1}$. & Independent generation
of $(S_t,U_t,K_t)$ and the pre-observation factorization
\eqref{eq:asym-A1b-factorization}. \\
(A2) & Preserve the current random-coding ensemble under history
conditioning. & Independent generation of the round codebooks, so
$\mathsf C_t\perp\mathscr H_{t-1}^{-t}$. \\
(A3) & Use the second identity to show that the corresponding conditional mutual information is zero, and use the Markov condition to ensure that the current channel accesses $U_t$ only through $C_t$. & The current-round
factorization in Lemma~\ref{lem:asym-round-local-factorization}. \\
(A4a) & Bound the one-round leakage term for $Q_t$ by $\delta_{n,t}$. &
Corollary~\ref{cor:history-conditioned-one-round} under the adaptive-round
identification \eqref{eq:asym-A4-explicit-map}. \\
(A4b) & Initialize the recursion at round~1. & The initialization
specialization \eqref{eq:asym-initialization-map}, followed by data
processing that discards the dummy protected component. \\
\hline
\end{tabular}
\end{table*}
Table~\ref{tab:condition-use-verification-map} separates the role of each
condition in the propagation proof from the concrete fact that verifies it.
In particular, the table does not rederive the one-round leakage bound.

\subsection{One-round leakage argument for the ideal process}
\begin{proposition}[Conditional one-time pad with side information]
\label{prop:abstract-otp}
Let $D$ be a conditioning variable.  Conditional on every $D=d$, let $U$
and $K$ be uniform on the same finite abelian group, let
$U\perp(K,W,E)\mid D$, and let $K\perp W\mid D$.  If $C=U\oplus K$, then
\begin{equation}
I(U;C,E\mid W,D)\leq I(K;E\mid W,D).
\label{eq:abstract-otp}
\end{equation}
\end{proposition}

The proposition is proved in Appendix~\ref{app:embedded-key-proofs}.  We now combine it with the
independence in (A1a) and the kernel conditions in (A3) in the order in which they enter
the one-step recursion.

\begin{lemma}[Prior key leakage bounds the encrypted-message leakage]
\label{lem:otp-leakage-absorption}
When the ideal process satisfies the independence condition (A1a), the
independence condition (A1b), and the
first identity and the Markov condition in (A3),
\begin{align}
&I(U_t;Z_t\mid Z^{1:t-1},W_{t-1},Q_t,
\boldsymbol{\mathsf C})\leq I(K_{t-1};Z^{1:t-1}
\mid W_{t-1},\boldsymbol{\mathsf C}).
\label{eq:otp-leakage-absorption}
\end{align}
\end{lemma}

\begin{lemma}[Leakage recursion from round $t-1$ to round $t$]
\label{lem:one-step-recursion}
With all codebook conditioning retained,
\begin{align}
I(W_t,K_t;Z^{1:t}\mid\boldsymbol{\mathsf C})\le 
& I(W_{t-1},K_{t-1};Z^{1:t-1}
\mid\boldsymbol{\mathsf C})\notag\\
&+I(Q_t;Z_t\mid W_{t-1},K_{t-1},
Z^{1:t-1},\boldsymbol{\mathsf C}).
\label{eq:one-step-recursion-revised}
\end{align}
\end{lemma}

The proof of Lemma \ref{lem:otp-leakage-absorption} and Lemma \ref{lem:one-step-recursion} can be found in Appendix \ref{app:otp-leakage-absorption-proof} and Appendix \ref{app:one-step-recursion-proof}, respectively.

\subsection{Transfer to the implemented process}
The preceding subsections complete the one-round leakage argument for the ideal true-key
process.  We now quantify the effect of replacing the true key by User~1's
decoded estimate in the implemented process and then select one deterministic codebook collection.

For distributions $P$ and $Q$ on the same finite alphabet, write
\begin{equation}
d_{\rm TV}(P,Q):=\frac12\lVert P-Q\rVert_1.
\label{eq:tv-definition}
\end{equation}

\begin{lemma}[Mutual-information continuity with a common protected marginal]
\label{lem:common-marginal-continuity}
Let $P_{WZ}$ and $Q_{WZ}$ have the same marginal on $W$, and put
$d=d_{\rm TV}(P_{WZ},Q_{WZ})$.  Then
\begin{equation}
I_P(W;Z)\leq I_Q(W;Z)+
2d\log|\mathcal W|+2h_2(\min\{d,1/2\}).
\label{eq:common-marginal-continuity}
\end{equation}
Here $h_2$ uses natural logarithms.
\end{lemma}
The implemented and ideal processes have the same marginal law on $W_T$, because
they use the same payload and key generation rules and differ only in the
transmitted codeword lookup.  Hence the common-protected-marginal
continuity lemma applies to their coupled output laws.

\begin{lemma}[Coupling bound for the implemented and ideal processes]
\label{lem:real-ideal-agreement}
Fix a codebook collection $\boldsymbol c$ and couple the implemented and ideal
processes by using the same payload and key variables, coding-randomization
indices, stochastic-prefix seeds, physical-channel randomness, and
deterministic decoders.  Define
\begin{equation}
\mathcal F_{n,T}(\boldsymbol c)
:=\bigcup_{t=2}^T\{\widehat K_{t-1}^{\rm real}\ne K_{t-1}\}.
\label{eq:real-ideal-failure-event}
\end{equation}
If the two histories agree through round $t-1$ and
$\widehat K_{t-1}^{\rm real}=K_{t-1}$, then their round-$t$ ciphertexts,
codeword lookups, prefix outputs, channel outputs, and decoder outputs agree.
Consequently, the final coupled pairs agree on
$\mathcal F_{n,T}(\boldsymbol c)^c$, and
\begin{equation}
d_{\rm TV}(P_{W_TZ^{1:T}}^{\rm real},
P_{W_TZ^{1:T}}^{\rm ideal})
\leq\Pr\{\mathcal F_{n,T}(\boldsymbol c)\}.
\label{eq:real-ideal-tv-coupling}
\end{equation}
The failure event is deliberately conservative: it includes a wrongly
supplied key even if the resulting ciphertext happens to coincide with the
ideal ciphertext.
\end{lemma}

\begin{proposition}[Transfer from the ideal process to the implemented process]
\label{prop:abstract-transfer}
Let the implemented process use a decoded key $\widehat K_{t-1}$ in round $t$,
whereas the ideal process uses $K_{t-1}$.  Use the coupling of
Lemma~\ref{lem:real-ideal-agreement}.  For a fixed codebook collection
$\boldsymbol c$, let
$p_{n,T}(\boldsymbol c)=\Pr\{\mathcal F_{n,T}(\boldsymbol c)\}$ and put
$\bar p_{n,T}=\mathbb E_{\boldsymbol{\mathsf C}}p_{n,T}
(\boldsymbol{\mathsf C})$.  Then
\begin{align}
\mathbb E_{\boldsymbol{\mathsf C}}
 I_{\rm real}(W_T;Z^{1:T}\mid\boldsymbol{\mathsf C})&\leq
\mathbb E_{\boldsymbol{\mathsf C}}
 I_{\rm ideal}(W_T;Z^{1:T}\mid\boldsymbol{\mathsf C})
+\eta_{n,T},
\label{eq:abstract-transfer}
\end{align}
where 
\begin{equation}\label{eq:abstract-transfer-eta}
	\eta_{n,T} :=2\bar p_{n,T}\log|\mathcal W_T| +2h_2(\min\{\bar p_{n,T},1/2\}).
\end{equation}
\end{proposition}

\begin{proposition}[Simultaneous codebook selection]
\label{prop:abstract-selection}
Suppose the ensemble-average implemented-process decoding error and ensemble-average
implemented-process leakage are bounded by $a_{n,T}$ and $b_{n,T}$, respectively.  Then
there is one deterministic codebook realization for which the error is at
most $2a_{n,T}$ and the leakage is at most $2b_{n,T}$.
\end{proposition}

The continuity lemma, the transfer proposition, and the simultaneous codebook-selection proposition are proved in
Appendix~\ref{app:embedded-key-proofs}.  Simultaneous codebook selection is applied
only after the ideal recursion and the transfer from the ideal process to the implemented process have both
been completed.

\subsection{Multiround propagation theorem}
\begin{theorem}[Embedded-key leakage propagation]
\label{thm:embedded-key-propagation}
Suppose that the ideal embedded-key process, the round-codebook ensemble,
and the associated current-round observation kernels defined above satisfy
conditions (A1a)--(A4b).  Then the ideal process
satisfies
\begin{equation}
\mathbb E_{\boldsymbol{\mathsf C}}
 I(W_T,K_T;Z^{1:T}\mid\boldsymbol{\mathsf C})
\leq\sum_{t=1}^T\delta_{n,t}.
\label{eq:abstract-propagation-with-key}
\end{equation}
Consequently, after discarding the terminal key and transferring to the
implemented decoded-key process, the ensemble-average leakage obeys the bound
\begin{equation}
\mathbb E_{\boldsymbol{\mathsf C}}
 I_{\rm real}(W_T;Z^{1:T}\mid\boldsymbol{\mathsf C})
\leq\sum_{t=1}^T\delta_{n,t}+\eta_{n,T}.
\label{eq:abstract-propagation-real}
\end{equation}
If the corresponding ensemble-average decoding-error bound is
$a_{n,T}$, there exists a deterministic implemented code with
\begin{align}
P_e&\leq2a_{n,T},\\
I(W_T;Z^{1:T})&\leq
2\sum_{t=1}^T\delta_{n,t}+2\eta_{n,T}.
\end{align}
\end{theorem}

\begin{IEEEproof}
The proof consists of the following 4 steps. 

\textit{Step 1: Ideal one-round recursion.}
Let
\begin{equation}
\Lambda_t:=\mathbb E_{\boldsymbol{\mathsf C}}
 I(W_t,K_t;Z^{1:t}\mid\boldsymbol{\mathsf C}).
\end{equation}
Initialization gives $\Lambda_1\leq\delta_{n,1}$.  For every $t\geq2$,
Lemma~\ref{lem:one-step-recursion} and
\eqref{eq:abstract-one-round-tower} give
$\Lambda_t\leq\Lambda_{t-1}+\delta_{n,t}$.

\textit{Step 2: Ideal multiround bound.}
Induction therefore yields
\eqref{eq:abstract-propagation-with-key}.  This is the point at which the
ideal true-key analysis is complete.

\textit{Step 3: Discarding the terminal key and transferring to the actual
process.}
Discarding $K_T$ by data processing gives the same upper bound
for the ideal leakage of $W_T$.  Proposition~\ref{prop:abstract-transfer}
then adds the correction $\eta_{n,T}$ and proves
\eqref{eq:abstract-propagation-real} for the implemented decoded-key
process.

\textit{Step 4: Simultaneous codebook selection.}
After obtaining the implemented-process ensemble bounds for decoding error and leakage, we apply Proposition~\ref{prop:abstract-selection}. This gives one deterministic
codebook collection satisfying the two displayed deterministic bounds.
\end{IEEEproof}

\begin{corollary}[Growing number of rounds]
\label{cor:abstract-growing-rounds}
Assume the stationary bounds
\allowdisplaybreaks
\begin{align}
\delta_{n,t}&\leq c_1e^{-n\beta},\\
\bar p_{n,T}&\leq c_2Te^{-n\alpha},\\
a_{n,T}&\leq c_3Te^{-n\gamma},\\
\log|\mathcal W_T|&\leq c_4nT
\end{align}
for constants $c_1,c_2,c_3,c_4,\alpha,\beta,\gamma>0$.  If $T=T_n\to\infty$ and
\begin{equation}
\frac{\log T_n}{n}\longrightarrow0,
\end{equation}
then the deterministic-code error, the propagated leakage, and the
ideal-to-implemented correction all vanish.  If only the first round is used for
initialization, the rate factor due to the initialization round $(T_n-1)/T_n$ converges to one.
Fixed-input padding may be appended to realize every sufficiently large
total blocklength without increasing error or leakage.
\end{corollary}

Theorem \ref{thm:embedded-key-propagation} establishes only the implication from the
conditions to multiround leakage control.  The following subsection shows 
that the asymmetric key-assisted adaptive code construction satisfies every one of those conditions.

\subsection{Verification for the asymmetric adaptive code construction}
\label{subsec:verify-A1-A4}

We now verify conditions \emph{(A1a), (A1b), and (A2)--(A4b)} for the
ideal process and round-codebook ensemble induced by the concrete asymmetric adaptive code 
construction of Section~\ref{sec:asymmetric-construction}.

\medskip
\noindent\textbf{Step 1: Verification of conditions (A1a) and (A1b).}
For every $t\geq2$, the variables
$S_t,U_t,M_2^t,K_t,O_{1,t},O_{2,t}$ are generated mutually
independently and uniformly before codebook lookup and independently of all
variables generated in preceding rounds and of the complete codebook
collection.  Hence \eqref{eq:abstract-current-freshness} holds.

Here $W_{t-1}$ contains only the payload source variables
$(S_2,U_2,\ldots,S_{t-1},U_{t-1})$; it contains no decoded key, channel
output, prefix randomness, or coding-randomization index.  The key $K_j$ is
generated in round $j$ and is first used only in round $j+1$ to encrypt
$U_{j+1}$.  In particular, $U_{t-1}$ uses $K_{t-2}$, not $K_{t-1}$, so
$K_{t-1}$ is not used to generate any component of $W_{t-1}$.  Passing the
independently generated sources through the channel does not change their
unconditioned source marginal; correlations arise only after a channel
output or decoding event is conditioned upon.  Therefore the independent
generation rule gives, before adjoining the past observation,
\begin{equation}
P_{W_{t-1},K_{t-1}\mid\boldsymbol{\mathsf C}}
=P_{W_{t-1}\mid\boldsymbol{\mathsf C}}
 P_{K_{t-1}\mid\boldsymbol{\mathsf C}},
\label{eq:asym-A1b-factorization}
\end{equation}
which is \eqref{eq:abstract-preobservation-closure}.  This statement is
made before adjoining $Z^{1:t-1}$ to the conditioning.  We make no claim
that the key and accumulated payload remain independent after the eavesdropper's past
observation is known.

\medskip
\noindent\textbf{Step 2: Verification of condition (A2).}
The round codebooks are generated independently.  Define
\begin{align}
\mathscr H_{t-1}^{-t}:=\sigma(&W_{t-1},K_{t-1},Z^{1:t-1},
\boldsymbol{\mathsf C}_{<t},\boldsymbol{\mathsf C}_{>t}).
\label{eq:asym-history}
\end{align}
Then $\mathsf C_t$ is independent of $\mathscr H_{t-1}^{-t}$.  Importantly,
$\mathsf C_t$ is not included in the history on which the one-round
ensemble bound is conditioned.  We first apply that bound with expectation
over $\mathsf C_t$ and then use the tower property to average over the
history and the remaining codebooks.

\medskip
\noindent\textbf{Step 3: Verification of condition (A3).}
Fix the history before round $t$.  The current primitive variables remain independent of the preceding variables, the current codebook retains its product ensemble law, and the
current output is generated only through the two current codeword lookups,
the stochastic prefixes, and the memoryless channel.  The following lemma
records this factorization and distinguishes the kernel with a fixed
ciphertext from the kernel obtained after averaging the history-wise uniform
ciphertext.  Keeping the statement separate also distinguishes structural
locality from the quantitative resolvability estimate used later.

\begin{lemma}[Round-local factorization for the asymmetric construction]
\label{lem:asym-round-local-factorization}
Fix $t\geq2$ and the ideal process, and let $\mathscr H_{t-1}^{-t}$ be as in
\eqref{eq:asym-history}.  Conditional on any history realization and on the
codebooks outside round $t$, the current codebook $\mathsf C_t$ remains
independent.  The six fresh variables
\begin{equation}
S_t,\quad U_t,\quad M_2^t,\quad K_t,\quad
O_{1,t},\quad O_{2,t}
\label{eq:asym-six-fresh-variables}
\end{equation}
are mutually independent and uniform.  The ideal ciphertext $C_t$ is
uniform and independent of
\begin{equation}
(Q_t,M_2^t,O_{1,t},O_{2,t},\mathscr H_{t-1}^{-t}).
\label{eq:asym-current-factorization}
\end{equation}
Write $\mathsf c_t$ for the current codebook realization, to distinguish it
from a ciphertext value $c$.  For a realization $q=(\sigma,\kappa)$ of
$Q_t=(S_t,K_t)$, the components $\sigma$ and $\kappa$ index the
User~1 and User~2 codewords, respectively.  For fixed history and
$\mathsf c_t$, the complete indices select auxiliary codewords
deterministically.  The stochastic prefixes are then applied through the
conditional memoryless kernels
\begin{equation}
P_{X_1^nX_2^n\mid V_1^nV_2^n}
=P_{X_1\mid V_1}^{\otimes n}P_{X_2\mid V_2}^{\otimes n}.
\label{eq:asym-prefix-kernel}
\end{equation}
Accordingly, the fixed-history current-round law factors as
\begin{align}
&P(q,c,m,o_1,o_2,v_1^n,v_2^n,x_1^n,x_2^n,z^n
 \mid h,\mathsf c_t)\notag\\
&\quad=P_Q(q)P_C(c)P_{M_2}(m)P_{O_1}(o_1)P_{O_2}(o_2)\notag\\
&\qquad\times
 \mathbf 1\{v_1^n=v_1^n(\sigma,c,o_1;\mathsf c_t)\}
 \mathbf 1\{v_2^n=v_2^n(\kappa,m,o_2;\mathsf c_t)\}\notag\\
&\qquad\times P_{X_1\mid V_1}^{\otimes n}(x_1^n\mid v_1^n)
 P_{X_2\mid V_2}^{\otimes n}(x_2^n\mid v_2^n)
 P_{Z\mid X_1X_2}^{\otimes n}(z^n\mid x_1^n,x_2^n).
\label{eq:A3-factorization}
\end{align}
Thus the two lookups use the appropriate components of the protected state
$q=(\sigma,\kappa)$ explicitly.  In particular, prefix outputs are not
asserted to be independent of the selected codewords.
Consequently,
\begin{align}
P_{Z_t\mid W_{t-1},Z^{1:t-1},Q_t,C_t,
\boldsymbol{\mathsf C}}
 &=P_{Z_t\mid Q_t,C_t,\mathsf C_t},
\label{eq:asym-round-locality-1}\\
P_{Z_t\mid W_{t-1},Z^{1:t-1},Q_t,
\boldsymbol{\mathsf C}}
 &=P_{Z_t\mid Q_t,\mathsf C_t},
\label{eq:asym-round-locality-2}
\end{align}
and
\begin{equation}
U_t-(C_t,Z^{1:t-1},W_{t-1},Q_t,
\boldsymbol{\mathsf C})-Z_t
\label{eq:asym-otp-markov}
\end{equation}
is a Markov chain.
\end{lemma}

\begin{IEEEproof}
Because $K_{t-1}$ is measurable with respect to
$\mathscr H_{t-1}^{-t}$, every admissible history realization $h$ fixes a
value $k(h)$ of that key.  Hence, for every ciphertext symbol $c$,
\begin{align}
\Pr\{C_t=c\mid\mathscr H_{t-1}^{-t}=h\}
&=\Pr\{U_t=c\ominus k(h)\}=|G_n|^{-1}.
\label{eq:asym-ciphertext-uniform-calculation}
\end{align}
Let $L_t=(S_t,K_t,M_2^t,O_{1,t},O_{2,t})$.  Freshness of $U_t$ and
independent generation of $L_t$ give, for every admissible $h$,
\begin{equation}
P_{C_tL_t\mid\mathscr H_{t-1}^{-t}=h}
=P_{C_t\mid\mathscr H_{t-1}^{-t}=h}
 P_{L_t\mid\mathscr H_{t-1}^{-t}=h}.
\label{eq:asym-history-product-uniformity}
\end{equation}
Thus the six current index coordinates required by the one-round
resolvability bound are conditionally mutually independent and uniform.  Independent generation gives
$\mathsf C_t\perp\mathscr H_{t-1}^{-t}$.  Current codeword lookup, the two
conditional prefix kernels in \eqref{eq:asym-prefix-kernel}, and the
memoryless channel then give \eqref{eq:A3-factorization}; none of these
kernels contains a past variable.
Marginalizing \eqref{eq:A3-factorization} over $M_2^t$, the two coding-randomization indices,
and both prefix outputs gives, for fixed $(q,c,\mathsf c_t)$,
\begin{align}
&P_{Z_t\mid q,c,\mathsf c_t}(z^n)\notag\\
&=\frac{1}{|\mathcal M_2||\mathcal O_1||\mathcal O_2|}
  \sum_{m,o_1,o_2}\sum_{x_1^n,x_2^n}
  P_{X_1\mid V_1}^{\otimes n}(x_1^n\mid
       v_1^n(\sigma,c,o_1;\mathsf c_t))\notag\\
&\quad\times P_{X_2\mid V_2}^{\otimes n}(x_2^n\mid
       v_2^n(\kappa,m,o_2;\mathsf c_t))
  P_{Z\mid X_1X_2}^{\otimes n}(z^n\mid x_1^n,x_2^n).
\label{eq:A3-explicit-kernel}
\end{align}
The right-hand side contains no past variable, proving
\eqref{eq:asym-round-locality-1}.  Averaging this display further over the
history-wise uniform ciphertext gives
\begin{equation}
P_{Z_t\mid q,\mathsf c_t}(z^n) 
=|G_n|^{-1}\sum_{c\in G_n}P_{Z_t\mid q,c,\mathsf c_t}(z^n),
\label{eq:A3-explicit-averaged-kernel}
\end{equation}
which proves \eqref{eq:asym-round-locality-2}.  Finally, conditional on
$(C_t,Q_t,\mathsf C_t)$ and the fresh averaging variables, current input
selection and the memoryless channel do not use $U_t$ separately.  After
marginalization this is exactly the Markov chain
\eqref{eq:asym-otp-markov}.
\end{IEEEproof}

Lemma~\ref{lem:asym-round-local-factorization} proves that the ideal
current-round observation kernel satisfies condition (A3).  In particular,
the first kernel keeps the ciphertext visible for the one-time-pad step,
whereas the averaged kernel removes the old-payload/current-output term.
The Markov chain is the precise condition used when current encrypted
payload leakage is transferred to the preceding key state.

\medskip
\noindent\textbf{Step 4: Verification of conditions (A4a) and (A4b).}
Conditional on every admissible realization of
$\mathscr H_{t-1}^{-t}$, Lemma~\ref{lem:asym-round-local-factorization}
shows that the ideal-round coordinates have the product-uniform law required
by Corollary~\ref{cor:history-conditioned-one-round}.  The explicit
identification is
\begin{align}
(A_1,J_1,B_1)&=(S_t,C_t,O_{1,t}),\notag\\
(A_2,J_2,B_2)&=(K_t,M_2^t,O_{2,t}).
\label{eq:asym-A4-explicit-map}
\end{align}
The current codebook is independent of the history by condition (A2).
Corollary~\ref{cor:history-conditioned-one-round} therefore gives, uniformly
in the history value $h$,
\begin{equation}
\mathbb E_{\mathsf C_t}
 I(Q_t;Z_t\mid\mathscr H_{t-1}^{-t}=h,\mathsf C_t)
\leq\delta_n(P,s).
\label{eq:asym-A4-bound}
\end{equation}
Averaging over the history and the codebooks outside round $t$ yields
\eqref{eq:abstract-one-round-tower}, proving (A4a).

For round~1, use the specialization
\eqref{eq:asym-initialization-map}.  The same one-round proposition protects
$(D_{1,\mathrm{sec}},K_1)$; data processing discards the dummy coordinate
and gives \eqref{eq:abstract-initial-bound}.  This proves (A4b) without a
separate initialization resolvability argument.

The four verification steps establish all hypotheses of
Theorem~\ref{thm:embedded-key-propagation}; applying it gives the following
coding bound for the concrete asymmetric construction.
\begin{corollary}[Fixed-round bound for the asymmetric construction]
\label{cor:fixed-round-asymmetric}
Fix $P\in\mathcal Q$, $s\in(0,1]$, and $T\geq2$.  For component rates
satisfying \eqref{eq:asym-key-size} and
\eqref{eq:asym-finite-s-conditions}, there exists a deterministic asymmetric
adaptive code of blocklength $nT$ and payload sizes $e^{n(T-1)R_i}$ such
that
\begin{align}
P_e^{nT}&\leq2a_{n,T},\\
I(W_T;Z^{1:T})&\leq2T\delta_n(P,s)+2\eta_{n,T},
\end{align}
where $a_{n,T}$, $\delta_n(P,s)$, and $\eta_{n,T}$ are given in
\eqref{eq:asym-ensemble-error}, \eqref{eq:asym-delta}, and
\eqref{eq:abstract-transfer-eta}, respectively.  For fixed $T$, both bounds
decay exponentially in $n$.  More explicitly, if
$\bar p_{n,T}\leq c_Te^{-n\alpha}$ and
$\log|\mathcal W_T|\leq c'_Tn$, then for every
$0<\alpha'<\alpha$ there is $c''_T<\infty$ such that
\begin{equation}
\eta_{n,T}\leq c''_T n e^{-n\alpha'}
\end{equation}
for all sufficiently large $n$; the polynomial and binary-entropy factors
are absorbed by the exponent margin.
\end{corollary}

\begin{IEEEproof}
The ideal embedded-key process, the round-codebook ensemble, and the
current-round observation kernels were verified above to satisfy conditions
(A1a), (A1b), and (A2)--(A4b).  Theorem~\ref{thm:embedded-key-propagation}
gives the ideal leakage recursion and its transfer to the implemented encoder.
Let $\mathcal E_{\rm key}$ denote use of a wrongly decoded key and $\mathcal E_{\rm nom}$ the event that at least one complete index is decoded incorrectly.  Since each key is a decoded complete-index component, $\mathcal E_{\rm key}\subseteq\mathcal E_{\rm nom}$; hence $\bar p_{n,T}\le a_{n,T}$.  After ensemble transfer, Proposition~\ref{prop:abstract-selection} is applied only to implemented-process error and implemented-process leakage.
\end{IEEEproof}

\begin{corollary}[Growing rounds and removal of initialization loss]
\label{cor:growing-rounds-asymmetric}
Suppose the finite-$s$ inequalities
\eqref{eq:asym-finite-s-conditions} hold with positive margins.  For total
blocklength $N$, choose
\begin{equation}
T_N=\lfloor N^{1/3}\rfloor,\qquad
n_N=\lfloor N/T_N\rfloor,
\end{equation}
run the active construction for $N'_N=n_NT_N$ channel uses, and append fixed
inputs for the remaining $N-N'_N$ uses.  Then error and leakage vanish, and
the effective payload rates converge to $(R_1,R_2)$.  No positive exponent
per total blocklength $N$ is asserted for this growing-round sequence.
\end{corollary}

\begin{IEEEproof}
Positive finite-$s$ margins give exponential bounds of the form required by
Corollary~\ref{cor:abstract-growing-rounds}.  The choices above satisfy
$T_N\to\infty$, $n_N\to\infty$, and $\log T_N/n_N\to0$.  Also,
$N'_N/N\to1$ and $(T_N-1)/T_N\to1$.  For each $N$, the simultaneous codebook-selection
argument is applied to that $N$-dependent collection of $T_N$ round
codebooks, yielding one deterministic collection for the corresponding
code.  By channel memorylessness, the fixed
padding outputs are independent of the payload conditional on the active
output and do not increase error or leakage.
\end{IEEEproof}

\medskip
\noindent\textbf{Transition to rate projection.}
The probabilistic analysis is complete.  Section~\ref{sec:region-geometry}
projects the key-size condition and the five strict finite-$s$ inequalities
in \eqref{eq:asym-finite-s-conditions} onto the payload rates.

\section{Exact rate projection and fixed-distribution geometry}
\label{sec:region-geometry}

Sections~\ref{sec:asymmetric-construction} and
\ref{sec:embedded-key-propagation} have established a concrete asymmetric
adaptive code with vanishing decoding error and information leakage whenever its
component rates satisfy the strict finite-$s$ conditions in
\eqref{eq:asym-finite-s-conditions} together with the key-size condition
\eqref{eq:asym-key-size}.  
The remaining task is a rate-projection problem.  In this section, we
project the component-rate system, connect the strict Shannon system to the
finite-$s$ bounds, and finally compare the projected adaptive polygon with
the non-adaptive rectangle.

Fix $P\in\mathcal Q_{\rm F}$ and abbreviate
\begin{align}
A=A(P),~ &B=B(P),~ \notag\\
C=C(P),~ &D=D(P),~ E=E(P).
\label{eq:geometry-abbreviations}
\end{align}
Recall that
\begin{equation}
\mathcal Q_{\rm F}
=\{P\in\mathcal Q:C<A,\ D<B,\ E<A+B\}.
\label{eq:geometry-QF}
\end{equation}

\subsection{Exact projection of the asymmetric component-rate system}

At the Shannon-information level, the reliability and one-round secrecy
conditions in \eqref{eq:asym-finite-s-conditions} become
\allowdisplaybreaks
\begin{align}
\widetilde R_1+\widetilde R_{1,r}&<A,
\label{eq:projection-rel1}\\
\widetilde R_2+\widetilde R_{2,r}&<B,
\label{eq:projection-rel2}\\
\widetilde R_{1,r}&>C,
\label{eq:projection-sec1}\\
\widetilde R_{2,r}+R_2&>D,
\label{eq:projection-sec2}\\
\widetilde R_{1,r}+\widetilde R_{2,r}+R_2&>E.
\label{eq:projection-sec12}
\end{align}
They are supplemented by the identities
\begin{align}
R_1&=R_{1,\mathrm{sec}}+R_{1,e},~
\widetilde R_1=R_{1,\mathrm{sec}},~
\widetilde R_{1,r}=R_{1,e}+R_{1,o},
\label{eq:projection-user1}\\
\widetilde R_2&=R_2+R_{2,k},~
\widetilde R_{2,r}=R_{2,o},~
R_{1,e}=R_{2,k},
\label{eq:projection-user2}
\end{align}
and by nonnegativity of all component rates.

\begin{proposition}[Projection of the adaptive component-rate constraints]
\label{prop:exact-adaptive-projection}
For fixed $P\in\mathcal Q_{\rm F}$, every strictly feasible payload pair with $R_1>0$ is characterized exactly by
\begin{align}
R_1&<A,
\label{eq:adaptive-strict-A}\\
R_2&<B,
\label{eq:adaptive-strict-B}\\
R_1&<A+B-E,
\label{eq:adaptive-strict-E}\\
R_1+R_2&<A+B-C.
\label{eq:adaptive-strict-C}
\end{align}
The points with $R_1=0$ are obtained as limits of positive-$R_1$ points.  Consequently, the closure of the full projection is
\begin{equation}
\mathcal R_{\rm A}(P):=
\left\{ (R_1,R_2)\in\mathbb R_+^2 
\left |
\begin{array}{l}
	R_2\leq B,\\
	R_1\leq A,\\
	R_1\leq A+B-E,\\
	R_1+R_2\leq A+B-C
\end{array}
\right.
\right\}.
\label{eq:adaptive-fixed-P-region}
\end{equation}
\end{proposition}

\begin{IEEEproof}
The proof is given in Appendix \ref{app:exact-adaptive-projection-proof}.
\end{IEEEproof}

\subsection{From the strict projection to an achievable region}

\begin{theorem}[Asymmetric adaptive achievable region]
\label{thm:adaptive-achievable-region}
For the asymmetric adaptive construction of
Section~\ref{sec:asymmetric-construction}, whose multiround reliability and
secrecy were established in Section~\ref{sec:embedded-key-propagation}, the
following payload-rate region is achievable under strong one-sided secrecy:
\begin{equation}
\mathcal R_{\rm A}
:=\overline{\operatorname{conv}}\!\left(
\bigcup_{P\in\mathcal Q_{\rm F}}\mathcal R_{\rm A}(P)
\right).
\label{eq:adaptive-overall-region}
\end{equation}
\end{theorem}

\begin{IEEEproof}
The proof consists of the following 4 steps. 

\textit{Step 1: Fixed-$P$ exact projection and component reconstruction.}
Fix $P\in\mathcal Q_{\rm F}$.  If $(R_1,R_2)$ is a positive-$R_1$
interior point of $\mathcal R_{\rm A}(P)$, then all four inequalities in
Proposition~\ref{prop:exact-adaptive-projection} hold with positive slack.
Additional details are given in Appendix \ref{app:adaptive-achievable-region-proof-step1}.

\textit{Step 2: Finite-$s$ realization and fixed-round bounds.}
The structural conditions required for leakage propagation impose no
additional rate inequalities at this stage: they were already verified in
Section~\ref{sec:embedded-key-propagation} from the construction itself.
Conditions (A4a)--(A4b) are supplied by the history-conditioned one-round
corollary and its initialization specialization.  Hence, once the reconstructed component rates satisfy
the five strict finite-$s$ inequalities, the fixed-round and growing-round
bounds of Section~\ref{sec:embedded-key-propagation} apply to the concrete
code.  For the reconstructed component rates, define the five positive
Shannon margins
\begin{align}
\mu_1&:=A-(\widetilde R_1+\widetilde R_{1,r}),&
\mu_2&:=B-(\widetilde R_2+\widetilde R_{2,r}),\notag\\
\nu_1&:=\widetilde R_{1,r}-C,&
\nu_2&:=\widetilde R_{2,r}+R_2-D,\notag\\
\nu_{12}&:=\widetilde R_{1,r}+\widetilde R_{2,r}+R_2-E.
\label{eq:adaptive-Shannon-margins}
\end{align}
By \eqref{eq:model-continuity}, there is a sufficiently small fixed
$s\in(0,1]$ for which the corresponding finite-$s$ margins remain
positive.  Equivalently, all five inequalities
\eqref{eq:asym-finite-s-conditions} hold.  Let
\begin{align}
\alpha:=s\min\{&A_s-(\widetilde R_1+\widetilde R_{1,r}),
B_s-(\widetilde R_2+\widetilde R_{2,r})\}>0,
\label{eq:adaptive-alpha}\\
\beta:=s\min\{&\widetilde R_{1,r}-C_s,
\widetilde R_{2,r}+R_2-D_s,\notag\\
&\widetilde R_{1,r}+\widetilde R_{2,r}+R_2-E_s\}>0.
\label{eq:adaptive-beta}
\end{align}
The one-round estimates then imply, for constants independent of $n$ and
$T$,
\begin{align}
a_{n,T}&\leq 2T e^{-n\alpha},
\label{eq:adaptive-anT-explicit}\\
\delta_n(P,s)&\leq \frac{3}{s}e^{-n\beta}.
\label{eq:adaptive-delta-explicit}
\end{align}
For fixed $T$, Corollary~\ref{cor:fixed-round-asymmetric} therefore gives
vanishing implemented-process decoding error and
\begin{equation}
I(W_T;Z^{1:T})\leq \frac{6T}{s} e^{-n\beta}+2\eta_{n,T}.
\label{eq:adaptive-fixed-round-expanded}
\end{equation}
The first term is the ideal-process leakage propagated over the rounds;
$\eta_{n,T}$ is the ideal-to-implemented correction caused by use of an
incorrectly decoded preceding key.  The latter is controlled by the same
complete-index decoding bound because
$\mathcal E_{\rm key}\subseteq\mathcal E_{\rm nom}$.

\textit{Step 3: Growing rounds, transfer correction, padding, and rates.}
For an arbitrary total blocklength $N$, set
\begin{equation}
T_N:=\lfloor N^{1/3}\rfloor,\qquad
n_N:=\lfloor N/T_N\rfloor,
\qquad N'_N:=n_NT_N.
\label{eq:adaptive-growing-parameters}
\end{equation}
Then
\begin{equation}
T_N\to\infty,\quad n_N\to\infty,\quad
\frac{\log T_N}{n_N}\to0,
\quad\frac{N'_N}{N}\to1.
\label{eq:adaptive-growing-limits}
\end{equation}
By \eqref{eq:adaptive-anT-explicit} and the inclusion of key-use failure
in failure to decode a complete index,
\begin{equation}
\bar p_{n_N,T_N}\leq 2T_Ne^{-n_N\alpha}\longrightarrow0.
\label{eq:adaptive-pbar-growing}
\end{equation}
Moreover, the accumulated payload alphabet satisfies
\begin{equation}
\log|\mathcal W_{T_N}|
=n_N(T_N-1)R_1+o(n_NT_N)
\leq c n_NT_N
\label{eq:adaptive-W-size}
\end{equation}
for some fixed $c$.  Hence the continuity correction obeys
\begin{align}
\eta_{n_N,T_N}
&=2\bar p_{n_N,T_N}\log|\mathcal W_{T_N}|
 +2h_2(\min\{\bar p_{n_N,T_N},1/2\})\notag\\
&=O\!\left(n_NT_N^2e^{-n_N\alpha}\right)
 +2h_2(\min\{2T_Ne^{-n_N\alpha},1/2\})\notag\\
 &\longrightarrow 0.
\label{eq:adaptive-eta-growing}
\end{align}
Also,
\begin{equation}
T_N\delta_{n_N}(P,s)
\leq \frac{3T_N}{s}e^{-n_N\beta}\longrightarrow0,
\qquad
a_{n_N,T_N}\longrightarrow0,
\label{eq:adaptive-growing-error-leakage}
\end{equation}
where the limits follow from $\log T_N/n_N\to0$.  Thus reliability,
ideal leakage, and the ideal-to-implemented correction all vanish for the
active construction of length $N'_N$.

Append fixed channel inputs for the remaining $N-N'_N$ uses.  The decoders
ignore the padding outputs.  Channel memorylessness makes those outputs
conditionally independent of the messages given the active outputs, so
padding changes neither the active decoding decision nor the leakage.  The
message sizes remain $e^{n_N(T_N-1)R_i}$ up to integer rounding, and the
rates measured per total blocklength satisfy
\begin{equation}
\frac{n_N(T_N-1)}{N}R_i
=\frac{N'_N}{N}\frac{T_N-1}{T_N}R_i
\longrightarrow R_i,
\qquad i=1,2.
\label{eq:adaptive-effective-rates}
\end{equation}
The construction therefore meets the formal achievability definition for
every sufficiently large total blocklength.  Because both $n_N$ and $T_N$
vary, we make no claim of a positive error or leakage exponent normalized
by the total blocklength $N$.

\textit{Step 4: Fixed-$P$ boundary, time sharing, and final closure.}
For a boundary point of $\mathcal R_{\rm A}(P)$, choose nonnegative
interior rate pairs converging to it.  For the $k$th pair, the order of
choices is: rate backoff, strict component reconstruction, a small fixed
$s_k$ preserving all five margins, and finally growing-round block
parameters large enough that error, leakage, initialization loss, and
padding loss are each at most $1/k$.  A diagonal sequence proves
achievability of the entire closed fixed-$P$ region.

The distribution $P$ may be chosen arbitrarily from
$\mathcal Q_{\rm F}$.  Finite time sharing remains within the adaptive code
class: concatenate finitely many adaptive constituent blocks, reset the
local encoder state at each deterministic boundary, and let each encoder
use only outputs from its current constituent block.  The overall decoder
applies the constituent decoders separately.  The constituent messages, local randomness, codebooks, and channel blocks
are mutually independent across schedule segments, and the public schedule
is deterministic.  The union bound controls the error.  Applying the chain
rule segment by segment, with each current message independent of all other
segments and their observations, bounds total leakage by the sum of the
constituent leakages, exactly as in
\eqref{eq:timesharing-leakage-sum}; the state reset and the restriction to
current-segment history ensure the conditional independence used in that
calculation.  This realizes the convex hull.  Finally, choose a sequence of
convex-hull points approaching any point in the outer closure and select
the constituent total blocklengths diagonally so that the rate error,
decoding error, and leakage all vanish.  This proves
\eqref{eq:adaptive-overall-region}.
\end{IEEEproof}

\section{Comparison of the three inner bounds}
\label{sec:three-inner-bound-comparison}

The preceding sections derived two strong-secrecy inner bounds: the
non-adaptive region $\mathcal R_{\rm N}$ and the asymmetric adaptive region
$\mathcal R_{\rm A}$.  This section places them on the same logical line as
the earlier weak one-sided inner bound.  The comparison has two distinct
levels.  The weak region is stated for product input distributions with
$V_i=X_i$, whereas the new strong regions permit general finite auxiliaries
and stochastic prefixes.  Accordingly, a fixed-product-input comparison is
made first, and only then are the general fixed-distribution and overall
comparisons stated.

\subsection{Weak versus strong non-adaptive inner bounds}

Let
\begin{equation}
P_X=P_{X_1}P_{X_2}
\label{eq:comparison-product-input}
\end{equation}
be a product input distribution, set $V_i=X_i$, and evaluate the common
quantities of \eqref{eq:model-Shannon} under this specialization:
\begin{align}
\begin{split}\label{eq:comparison-product-quantities}
	A_X:=&I(Y_2;X_1\mid X_2),\\ 
	B_X:=&I(Y_1;X_2\mid X_1), \\
	C_X:=&I(Z;X_1),\\
	D_X:=&I(Z;X_2),\\
	E_X:=&I(Z;X_1,X_2).
\end{split}
\end{align}
Thus $(A_X,B_X,C_X,D_X,E_X)$ is exactly
$(A(P_X),B(P_X),\allowbreak C(P_X),D(P_X),E(P_X))$ with identity prefixes.

For this product input, the earlier weak one-sided inner bound \cite[Theorem 1]{QCHT2016} can be rewritten as
\begin{align}
\mathcal R_{\rm W}(P_X):=
\Bigl\{(R_1,R_2)&\in\mathbb R_+^2:
R_2\leq B_X,\quad R_1\leq A_X-C_X,\notag\\
& R_1\leq A_X+R_2-E_X\Bigr\}.
\label{eq:comparison-weak-region}
\end{align}
The corresponding specialization of the strong non-adaptive region is
\begin{align}
\mathcal R_{\rm N}(P_X):=
\Bigl\{(R_1,R_2)&\in\mathbb R_+^2:
R_2\leq B_X,\quad R_1\leq A_X-C_X,\notag\\
& R_1\leq A_X+B_X-E_X\Bigr\},
\label{eq:comparison-strong-product-region}
\end{align}
provided the strict feasibility conditions
\begin{equation}
C_X<A_X,\qquad D_X<B_X,
\qquad E_X<A_X+B_X
\label{eq:comparison-product-feasibility}
\end{equation}
are satisfied.  These are the feasibility conditions required by the
strong-secrecy construction used here.  They are imposed explicitly for this
comparison and are not inferred merely from nonnegativity of a rate pair in
$\mathcal R_{\rm W}(P_X)$.

\begin{proposition}[Weak-to-strong product-input inclusion]
\label{prop:weak-strong-product-inclusion}
For every product input $P_X$ satisfying
\eqref{eq:comparison-product-feasibility},
\begin{equation}
\mathcal R_{\rm W}(P_X)
\subseteq\mathcal R_{\rm N}(P_X).
\label{eq:comparison-fixed-product-inclusion}
\end{equation}
Consequently, after taking unions over strictly feasible product inputs,
finite convex hulls, and closure, the resulting weak construction-specific
inner bound is contained in the corresponding product-input specialization
of the strong non-adaptive inner bound.
\end{proposition}

\begin{IEEEproof}
The $R_2$ and individual $R_1$ inequalities are identical.  If
$(R_1,R_2)\in\mathcal R_{\rm W}(P_X)$, then
\begin{equation}
R_1\leq A_X+R_2-E_X
\leq A_X+B_X-E_X,
\end{equation}
which is the remaining strong non-adaptive inequality.  The operations of
union, convexification, and closure preserve inclusion.
\end{IEEEproof}

The feasibility qualifier is essential.   
The closed expression in
\eqref{eq:comparison-strong-product-region} is invoked only when the
underlying strict randomization-rate system is nonempty.  Closure of a
projected formula does not create operationally achievable points for a
fixed input whose strict auxiliary-rate system is empty.

\subsection{Strong non-adaptive versus asymmetric adaptive bounds}

We now return to a general auxiliary distribution $P\in\mathcal Q_{\rm F}$,
which may contain nontrivial $V_i$ and stochastic prefixes.  

Recall the fixed-distribution non-adaptive region
$\mathcal R_{\rm N}(P)$ defined in \eqref{eq:nonadaptive-region} and the fixed-distribution adaptive region $\mathcal R_{\rm A}(P)$  defined in \eqref{eq:adaptive-fixed-P-region}.  

With the
abbreviations as given in \eqref{eq:geometry-abbreviations}, $\mathcal R_{\rm N}(P)$ is a rectangle whose upper boundary is
$\min\{A-C,A+B-E\}$.  
Define
\begin{align}
	r_{\rm N}^{\max}:=&\min\{A-C,A+B-E\}\notag\\
	=&A-\max\{C,E-B\},
	\label{eq:geometry-N}\\
	r_{\rm A}^{\max}:=&\min\{A,A+B-E\}=A-\max\{0,E-B\},
	\label{eq:geometry-M}\\
	s_{\max}:=&A+B-C.
	\label{eq:geometry-S}
\end{align}
Then
\begin{align}
	\mathcal R_{\rm N}(P)
	=\{(R_1,R_2)\in\mathbb R_+^2:&R_2\leq B,\ R_1\leq r_{\rm N}^{\max}\},
	\label{eq:geometry-RN-compact}\\
	\mathcal R_{\rm A}(P)
	=\{(R_1,R_2)\in\mathbb R_+^2:&R_2\leq B,\ R_1\leq r_{\rm A}^{\max},\notag\\
	&R_1+R_2\leq s_{\max}\}.
	\label{eq:geometry-RA-compact}
\end{align}

The quantities $r_{\rm N}^{\max},r_{\rm A}^{\max},s_{\max}$ have a direct geometric meaning. As illustrated in Figure~\ref{fig:fixed-P-geometry}, the non-adaptive
region has height $r_{\rm N}^{\max}$ and width $B$.  The adaptive region has initial height
$r_{\rm A}^{\max}$, is capped by the sum-rate half-space $R_1+R_2\leq s_{\max}$, and has the same
right-hand endpoint $(B,r_{\rm N}^{\max})$ whenever the inclusion is strict.  The vertical
difference $G=r_{\rm A}^{\max}-r_{\rm N}^{\max}$ is the maximum User~1 rate gain when $R_2$ is fixed.  %Figure~\ref{fig:fixed-P-geometry} summarizes this geometry.

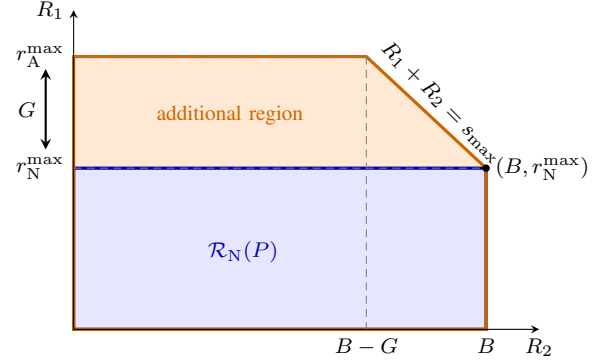
\begin{figure}[t]
	\centering
	\begin{tikzpicture}[x=0.88cm,y=0.82cm,>=stealth,font=\footnotesize]
		\coordinate (O) at (0,0);
		\coordinate (BN) at (6.2,2.6);
		\coordinate (NM) at (0,2.6);
		\coordinate (MM) at (0,4.4);
		\coordinate (K) at (4.4,4.4);
		\coordinate (B0) at (6.2,0);
		\coordinate (SL) at (5.34,3.46);
		
		\fill[blue!10] (O)--(B0)--(BN)--(NM)--cycle;
		\fill[orange!17] (NM)--(MM)--(K)--(BN)--cycle;
		\draw[blue!70!black,very thick] (O)--(B0)--(BN)--(NM)--cycle;
		\draw[orange!80!black,very thick] (O)--(B0)--(BN)--(K)--(MM)--cycle;
		
		\draw[->] (0,0)--(7.0,0) node[below] {$R_2$};
		\draw[->] (0,0)--(0,5.15) node[left] {$R_1$};
		\draw[densely dashed,gray] (NM)--(BN);
		\draw[densely dashed,gray] (MM)--(K);
		\draw[densely dashed,gray] (K)--(4.4,0);
		
		\node[left] at (NM) {$r_{\rm N}^{\max}$};
		\node[left] at (MM) {$r_{\rm A}^{\max}$};
		\node[below] at (B0) {$B$};
		\node[below] at (4.4,0) {$B-G$};
		\node[above,rotate=-45] at (SL) {$R_1+R_2=s_{\max}$};
		\node[blue!70!black] at (2.55,1.25) {$\mathcal R_{\rm N}(P)$};
		\node[orange!80!black] at (2.35,3.45) {additional region};
		
		\draw[<->,thick] (-0.42,2.9)--(-0.42,4.2)
		node[midway,left] {$G$};
		\fill (BN) circle (1.4pt);
		\node[right] at (BN) {$(B,r_{\rm N}^{\max})$};
	\end{tikzpicture}
	\caption{Fixed-distribution geometry under \eqref{eq:strictness-condition}.
		The non-adaptive region is the blue rectangle, and the adaptive region is
		the enclosing polygon.  The maximum User~1 rate gain at fixed $R_2$ is $G$.  This
		gain remains constant up to $R_2=B-G$ and then decreases to zero at
		$R_2=B$.  The figure is schematic.}
	\label{fig:fixed-P-geometry}
\end{figure}

\begin{theorem}[Fixed-distribution strictness criterion]
	\label{thm:fixed-P-gain}
	For every $P\in\mathcal Q_{\rm F}$,
	\begin{equation}\label{eq:comparison-strong-adaptive-inclusion}
		\mathcal R_{\rm N}(P)\subseteq\mathcal R_{\rm A}(P).
	\end{equation}
	This inclusion is strict if and only if
	\begin{equation}
		C(P)>0\quad\text{and}\quad E(P)<B(P)+C(P).
		\label{eq:strictness-condition}
	\end{equation}
\end{theorem}

\begin{IEEEproof}
	The proof is given in Appendix \ref{app:fixed-P-gain-proof}.
\end{IEEEproof}

The theorem is a pointwise comparison for a common auxiliary distribution
and common stochastic prefixes.  It identifies the local mechanism by which
adaptation enlarges this particular pair of construction-specific regions.
By itself, it does not imply strict inclusion after taking unions,
time sharing, and closure over all $P$.  

\begin{proposition}[User~1 rate gain at fixed $R_2$]
	\label{prop:fixed-P-gain-profile}
	Suppose \eqref{eq:strictness-condition} holds.  Define
	\begin{align}
		G(P):=&r_{\rm A}^{\max}-r_{\rm N}^{\max}\notag\\
		=&\min\{C(P),B(P)+C(P)-E(P)\}>0.
		\label{eq:fixed-P-max-gain}
	\end{align}
	For $r\in[0,B(P)]$, let $U_{\rm N}(r)$ and $U_{\rm A}(r)$ denote the largest
	$R_1$ allowed by the two fixed-distribution regions at $R_2=r$.  Then
	\begin{align}
		U_{\rm N}(r)&=r_{\rm N}^{\max},\\
		U_{\rm A}(r)&=\min\{r_{\rm A}^{\max},s_{\max}-r\},
	\end{align}
	and
	\begin{equation}
		U_{\rm A}(r)-U_{\rm N}(r)
		=\min\{G(P),B(P)-r\}.
		\label{eq:gain-profile}
	\end{equation}
	The two regions have the same maximum sum-rate, namely
	\begin{equation}
		A(P)+B(P)-\max\{C(P),E(P)-B(P)\}.
		\label{eq:common-sum-rate-journal}
	\end{equation}
\end{proposition}

\begin{IEEEproof}
	Under \eqref{eq:strictness-condition}, $r_{\rm N}^{\max}=A-C$ and
	\begin{align}
		r_{\rm A}^{\max}-r_{\rm N}^{\max}
		&=C-\max\{0,E-B\}\\
		&=\min\{C,B+C-E\},
	\end{align}
	which gives \eqref{eq:fixed-P-max-gain}.  For fixed $r$, subtracting $r_{\rm N}^{\max}$
	from the adaptive upper boundary gives
	\begin{align}
		U_{\rm A}(r)-r_{\rm N}^{\max}
		&=\min\{r_{\rm A}^{\max}-r_{\rm N}^{\max},s_{\max}-r-r_{\rm N}^{\max}\}\\
		&=\min\{G(P),B-r\},
	\end{align}
	because $s_{\max}-r_{\rm N}^{\max}=B$.  This proves \eqref{eq:gain-profile}.  The non-adaptive
	upper-right corner has sum-rate $r_{\rm N}^{\max}+B$, and maximizing over the adaptive
	polygon gives the same value.  Substitution of $r_{\rm N}^{\max}$ yields
	\eqref{eq:common-sum-rate-journal}.
\end{IEEEproof}

\begin{remark}[Interpretation and scope]
	At a fixed $P$, the additional region is controlled by two margins.  The
	quantity $C(P)$ is the penalty removed from User~1's non-adaptive individual
	bound, while $B(P)+C(P)-E(P)$ is the remaining sum-rate margin.  Their minimum
	is the maximum User~1 rate gain at a fixed User~2 rate.  This description is useful for
	understanding the geometry and for locating candidate distributions.  Nevertheless,  the
	fixed-$P$ gain should not be interpreted as an operational advantage over the
	optimized non-adaptive inner bound.  Section~\ref{sec:binary-family}
	provides an example where the adaptive gain survives optimization and time sharing.
\end{remark}

The domains in the two subsections must not be conflated.  The relation
\eqref{eq:comparison-fixed-product-inclusion} compares two regions at the
same product input with identity auxiliaries.  Equation
\eqref{eq:comparison-strong-adaptive-inclusion} compares the two new strong
regions at an arbitrary common $P\in\mathcal Q_{\rm F}$.  In particular,
the latter permits distributions not represented in the former
product-input family.

\subsection{Fixed-distribution and overall comparisons}

Recall the overall construction-specific regions $\mathcal R_{\rm N}$ and $\mathcal R_{\rm A}$, defined in \eqref{eq:nonadaptive-overall} and \eqref{eq:adaptive-overall-region}, respectively.
Because the fixed-$P$ inclusion holds for every feasible $P$ as shown in \eqref{eq:comparison-strong-adaptive-inclusion}, we have
\begin{equation}
\mathcal R_{\rm N}\subseteq\mathcal R_{\rm A}.
\label{eq:comparison-overall-inclusion}
\end{equation}
This implication uses only monotonicity of union, convex hull, and closure.
Strictness is different.  An adaptive point added at one distribution can
be covered by a non-adaptive region generated by another distribution or
by time sharing.  Therefore the fixed-$P$ condition
\eqref{eq:strictness-condition} does not by itself prove strict overall
inclusion.

Section~\ref{sec:binary-family} next proves overall strictness by
separating the two closed convex regions on the $R_1$ axis.  Equation~\eqref{eq:comparison-overall-inclusion} compares the two inner
bounds constructed in this paper.  The binary-family theorem is logically
stronger in a different direction: it supplies a converse for the full
non-adaptive stochastic code class defined in
Section~\ref{sec:model-framework}.

\begin{remark}[Overall optimization and auxiliary alphabets]
Because $\mathcal R_{\rm N}$ and $\mathcal R_{\rm A}$ are closed and convex,
strict overall inclusion could equivalently be certified by a strict
support-function inequality in some direction.  The binary family below
uses the simpler $R_1$-axis direction and proves a converse for the full
non-adaptive code class.  The unions defining the general inner bounds range
over finite auxiliary alphabets, and their displayed outer closures include
limits of such finite-alphabet choices and deterministic time sharing.  No
general cardinality reduction is needed for the binary separation, whose
converse is a compact optimization over two Bernoulli marginals.
\end{remark}
\section{A binary channel family with an overall adaptive gain}
\label{sec:binary-family}
This section gives an open two-parameter family of binary channels for
which adaptive coding achieves a rate pair unreachable by any non-adaptive
stochastic code under strong one-sided secrecy.  We first prove an operational non-adaptive converse and then derive
the adaptive separation as a corollary,
and then identify the complete non-adaptive capacity region for a
representative channel.

\subsection{Channel family and analytic overall separation}
\label{subsec:binary-analytic-separation}
Consider binary inputs and noiseless cross-links
\begin{equation}
Y_1=X_2,\qquad Y_2=X_1.
\end{equation}
The inequalities $0<a<b<1/2$ define an open subset of the two-dimensional
$(a,b)$ parameter plane.  For parameters in this set, let the eavesdropper's
binary-output channel be
\begin{equation}\label{eq:eve-family}
\begin{array}{c|cccc}
(x_1,x_2)&(0,0)&(0,1)&(1,0)&(1,1)\\ \hline
P_{Z\mid X_1X_2}(1\mid x_1,x_2)&a&b&1-b&1-a.
\end{array}
\end{equation}
This family is invariant under the simultaneous complement
$(x_1,x_2,z)\mapsto(1-x_1,1-x_2,1-z)$.  It is not, in general, an additive
binary symmetric channel: in the representation $Z=X_1\oplus N$, the
crossover law depends on whether $X_1=X_2$.  When $X_2$ is uniform, however,
the marginal channel from $X_1$ to $Z$ is a binary symmetric channel with
crossover probability $t=(a+b)/2$.

The proof uses an axis separation.  The symmetric product input gives the
adaptive point $(\log 2,0)$.  For the entire non-adaptive code class defined in
Section~\ref{sec:model-framework}, the
$R_1$-axis intercept is bounded by a compact optimization of
$H(X_1\mid Z)$ and is uniformly smaller than $\log 2$.  Nonnegativity then
shows that convexification cannot repair this axis gap.

\begin{lemma}[Removal of independent decoder side information]
\label{lem:remove-independent-side-information}
Let $(M,X)\perp S$.  Then
\begin{equation}
P_{M\mid X,S}=P_{M\mid X}
\label{eq:independent-side-information-posterior}
\end{equation}
for every pair $(x,s)$ of positive probability.  Consequently, the MAP
error based on $(X,S)$ equals the MAP error based on $X$ alone.  In
particular, if a prescribed decoder using $(X,S)$ has error probability
$\epsilon$, then an $X$-only decoder exists with error probability at most
$\epsilon$.
\end{lemma}
\begin{IEEEproof}
Independence gives
$P_{MXS}=P_{MX}P_S$, and division by $P_{XS}=P_XP_S$ proves
\eqref{eq:independent-side-information-posterior}.  The two MAP rules
therefore maximize the same posterior distribution.  The MAP error is no
larger than the error of any prescribed decoder.
\end{IEEEproof}

The converse below applies to the entire non-adaptive stochastic code class
of Section~\ref{sec:model-framework}: private local seeds are independent,
shared or correlated encoder randomness is excluded, and the public code
design and a deterministic time-sharing schedule may be common knowledge.
Within each encoder, however, the coordinates of $X_i^n$ may have arbitrary
temporal dependence induced by the message and private randomness; the
converse is not restricted to the memoryless random-coding construction used
for the inner bound.

\begin{theorem}[Non-adaptive converse for the binary channel family]
\label{thm:binary-nonadaptive-converse}
Let $\mathcal C_{\rm N}^{(1)}$ denote the capacity region under strong
one-sided secrecy when both encoders are restricted to be non-adaptive.
For independent random variables $X_1\sim\mathrm{Bern}(p_1)$ and
$X_2\sim\mathrm{Bern}(p_2)$, with $Z$ generated by
\eqref{eq:eve-family}, define
\begin{align}
M_{\rm N}(a,b)&:=\max_{p_1,p_2\in[0,1]}H(X_1\mid Z),
\label{eq:binary-MN-general}\\
\Delta(a,b)&:=\log 2-M_{\rm N}(a,b).
\label{eq:axis-gap-definition}
\end{align}
Then $\Delta(a,b)>0$, and every non-adaptive achievable pair satisfies
\begin{equation}
\mathcal C_{\rm N}^{(1)}
\subseteq [0,M_{\rm N}(a,b)]\times[0,\log 2].
\label{eq:binary-nonadaptive-capacity-outer}
\end{equation}
\end{theorem}
Thus the only excluded cross-encoder resource is shared or correlated
randomness; arbitrary within-encoder temporal correlation is allowed.
\begin{IEEEproof}
Consider an arbitrary sequence of non-adaptive stochastic codes.  Let its blocklength-$n$ message rates be
$R_{1,n},R_{2,n}$, its decoding error be $\epsilon_n$, and its leakage be
$\tau_n=I(M_1;Z^n)$.  Non-adaptivity, independent messages, and independent
local encoder randomness imply
\begin{equation}
(M_1,X_1^n)\perp(M_2,X_2^n).
\label{eq:binary-nonadaptive-product-sequences}
\end{equation}
Since $Y_2^n=X_1^n$, the prescribed User~2 decoder uses
$(M_2,X_2^n,X_1^n)$.  Apply
Lemma~\ref{lem:remove-independent-side-information} with
$M=M_1$, $X=X_1^n$, and $S=(M_2,X_2^n)$.  Equation
\eqref{eq:binary-nonadaptive-product-sequences} gives
\begin{equation}
P_{M_1\mid X_1^n,M_2,X_2^n}=P_{M_1\mid X_1^n}.
\label{eq:binary-decoder-side-information-removal}
\end{equation}
Thus the MAP error based on $X_1^n$ alone equals that based on all the
prescribed side information and is no larger than the error of the given
decoder.  Fano's inequality therefore gives
\begin{equation}
H(M_1\mid X_1^n)\leq 1+\epsilon_n\log|\mathcal M_{1,n}|.
\label{eq:binary-Fano-X1n}
\end{equation}
Using strong one-sided secrecy and conditioning reducing entropy, we obtain
\begin{align}
\log|\mathcal M_{1,n}|
&=I(M_1;Z^n)+I(M_1;X_1^n\mid Z^n)
  +H(M_1\mid X_1^n,Z^n)\notag\\
&\leq \tau_n+H(X_1^n\mid Z^n)
  +1+\epsilon_n\log|\mathcal M_{1,n}|.
\label{eq:binary-block-converse}
\end{align}
The conditional chain rule yields
\begin{align}
H(X_1^n\mid Z^n)
&=\sum_{j=1}^n H(X_{1,j}\mid X_1^{j-1},Z^n)\notag\\
&\leq\sum_{j=1}^n H(X_{1,j}\mid Z_j).
\label{eq:binary-single-letter-entropy-bound}
\end{align}
For every $j$, \eqref{eq:binary-nonadaptive-product-sequences} implies
$X_{1,j}\perp X_{2,j}$.  Thus the $j$th summand is evaluated at some
product Bernoulli input and is at most $M_{\rm N}(a,b)$.  Dividing
\eqref{eq:binary-block-converse} by $n$, using
$\epsilon_n\to0$, $\tau_n\to0$, and taking the limit gives
\begin{equation}
R_1\leq M_{\rm N}(a,b).
\label{eq:axis-wiretap-bound}
\end{equation}
For User~1's decoding of $M_2$, the observation is
$Y_1^n=X_2^n$.  Because the prescribed decoder has error at most
$\epsilon_n$, Fano's inequality gives
\begin{equation}
H(M_2\mid M_1,X_1^n,X_2^n)
\leq 1+\epsilon_n\log|\mathcal M_{2,n}|.
\end{equation}
Since $M_2\to X_2^n\to(M_1,X_1^n)$ is a Markov chain under the
non-adaptive product encoder law,
\begin{equation}
H(M_2\mid X_2^n)=H(M_2\mid M_1,X_1^n,X_2^n).
\label{eq:binary-M2-Markov-removal}
\end{equation}
Therefore
\begin{align}
(1-\epsilon_n)\log|\mathcal M_{2,n}|
&\leq I(M_2;X_2^n)+1\notag\\
&\leq H(X_2^n)+1\leq n\log 2+1,
\end{align}
which yields $R_2\leq\log 2$.  This proves
\eqref{eq:binary-nonadaptive-capacity-outer}.

It remains to show that the gap is positive.  For fixed 
$p_2=P(X_2=1)$, let $p_1=P(X_1=1)$ and
\begin{align}
q_0&:=P(Z=1\mid X_1=0)=(1-p_2)a+p_2b,\\
q_1&:=P(Z=1\mid X_1=1)=(1-p_2)(1-b)+p_2(1-a).
\end{align}
Then $q_0\in[a,b]$ and $q_1\in[1-b,1-a]$.  These intervals are disjoint,
so $q_0\ne q_1$ for every $p_2$.  If $p_1$ is $0$ or $1$, then
$H(X_1\mid Z)=0$.  If $0<p_1<1$, the distinct conditional output laws imply
$I(X_1;Z)>0$, and hence
\begin{equation}
H(X_1\mid Z)=h_2(p_1)-I(X_1;Z)<\log 2.
\end{equation}
The function $H(X_1\mid Z)$ is continuous on the compact square
$[0,1]^2$.  Its maximum is therefore strictly below $\log 2$, proving
$\Delta(a,b)>0$.
\end{IEEEproof}

\begin{corollary}[Operational adaptive separation]
\label{cor:binary-operational-separation}
For every channel in the family \eqref{eq:eve-family},
\begin{equation}
(\log 2,0)\in\mathcal R_{\rm A}
\setminus\mathcal C_{\rm N}^{(1)}.
\label{eq:binary-operational-separation}
\end{equation}
Consequently, adaptive coding achieves a rate pair unreachable by every
non-adaptive stochastic code in the code class of
Section~\ref{sec:model-framework}.
\end{corollary}
\begin{IEEEproof}
Choose $V_i=X_i$ and let  $X_1,X_2$ be independent and uniform.  Recall the
abbreviations as given in \eqref{eq:geometry-abbreviations}. We have here $A=B=\log 2$.  Complement
symmetry makes $Z$ uniform.  Writing $h_2$ for binary entropy with natural
logarithms and setting $t=(a+b)/2$, direct calculation gives
\begin{align}
C&=\log 2-h_2(t),\\
D&=\log 2-h_2(t),\\
E&=\log 2-\frac{h_2(a)+h_2(b)}2.
\end{align}
Because $0<a<b<1/2$, we have $0<C<A$, $0<D<B$, and $E<A+B$.
Moreover,
\begin{align}
A+B-E&=\log 2+\frac{h_2(a)+h_2(b)}2>\log 2,\\
A+B-C&=\log 2+h_2(t)>\log 2.
\end{align}
Thus all defining inequalities of $\mathcal R_{\rm A}(P)$ hold at
$(R_1,R_2)=(\log 2,0)$.
Hence $(\log 2,0)\in\mathcal R_{\rm A}$.  On the other hand,
Theorem~\ref{thm:binary-nonadaptive-converse} gives
$R_1\leq M_{\rm N}(a,b)<\log 2$ for every non-adaptive achievable pair.
Therefore $(\log 2,0)\notin\mathcal C_{\rm N}^{(1)}$, which proves
\eqref{eq:binary-operational-separation}.
\end{IEEEproof}
\begin{remark}[Relation to additive-channel converses]
The channel \eqref{eq:eve-family} is not an additive-noise TW-WC with input-independent
eavesdropper noise unless $a=b$.  Therefore the additive-channel converse
conditions used in earlier work do not apply directly.  The converse above
instead uses the product structure of arbitrary non-adaptive stochastic
encoders and the noiseless legitimate cross-links.  It is specific to
one-sided secrecy and does not assert a converse for arbitrary adaptive
codes.
\end{remark}

\subsection{Representative channel: analytic maximization and exact non-adaptive region}
\label{subsec:binary-numerical-illustration}
The strict overall separation in Corollary~\ref{cor:binary-operational-separation} is
entirely analytic.  For the representative channel with $a=1/10$ and $b=1/5$ as defined in \eqref{eq: representative channel}, the analysis below first
bounds the complete non-adaptive capacity region in terms of the scalar
quantity $M_{\rm N}$.  Appendix~\ref{app:binary-analytic-maximization}
then identifies $M_{\rm N}=h_2(3/20)$ by localizing the stationary maximizer
and proving strict concavity of the optimized one-dimensional profile.  The
reverse inclusion is supplied by the uniform product input.  We first record
the general one-dimensional reduction.

Let
\begin{equation}
 p:=P_{X_1}(1),\quad 
 t:=\frac{a+b}{2},\quad 
 d:=1-a-b=1-2t,
\end{equation}
and introduce
\begin{equation}
 u:=P_{Z\mid X_1}(1\mid0)
   =a+(b-a)P_{X_2}(1).
\end{equation}
Then $u\in[a,b]$ and
$P_{Z\mid X_1}(1\mid1)=u+d$.  Writing $h_2$ for binary entropy with
natural logarithms, we obtain
\begin{align}
 F_t(p,u)
 &: = H(X_1\mid Z)\notag\\
 &=h_2(p)+(1-p)h_2(u)+p h_2(u+d)-h_2(u+pd).
 \label{eq:numerical-profile-F}
\end{align}
Consequently,
\begin{equation}
 \max_{p_1,p_2\in[0,1]}H(X_1\mid Z)
 =\max_{\substack{0\leq p\leq1\\a\leq u\leq b}}F_t(p,u).
 \label{eq:numerical-two-variable-problem}
\end{equation}
The change from $(p_1,p_2)$ to $(p,u)$ is bijective because
$u=a+(b-a)p_2$ maps $[0,1]$ onto $[a,b]$.

\begin{proposition}[Reduction to a one-dimensional maximization]
\label{prop:numerical-profile-reduction}
For every fixed $u\in[a,b]$, the function
$p\mapsto F_t(p,u)$ is strictly concave on $(0,1)$ and has a unique
maximizer $p_t^\star(u)\in(0,1)$.  Therefore,
\begin{equation}
 \max_{\substack{0\leq p\leq1\\a\leq u\leq b}}F_t(p,u)
 =\max_{a\leq u\leq t}G_t(u),
 \quad
 G_t(u):=F_t(p_t^\star(u),u).
 \label{eq:numerical-one-variable-problem}
\end{equation}
Moreover,
\begin{equation}
 p_t^\star(2t-u)=1-p_t^\star(u),
 \quad
 G_t(2t-u)=G_t(u).
 \label{eq:numerical-profile-symmetry}
\end{equation}
\end{proposition}
\begin{IEEEproof}
Put $r=u+pd$.  Direct differentiation gives
\begin{equation}
 \frac{\partial^2F_t}{\partial p^2}
 =-\frac{1}{p(1-p)}+\frac{d^2}{r(1-r)}.
 \label{eq:numerical-Fpp}
\end{equation}
The identity
\begin{align}
&r(1-r)-d^2p(1-p)\notag\\
&\quad=(1-p)u(1-u)+p(u+d)(1-u-d)>0
\label{eq:numerical-variance-identity}
\end{align}
shows that \eqref{eq:numerical-Fpp} is strictly negative in the interior.
Also, $\partial F_t/\partial p$ tends to $+\infty$ as $p\downarrow0$ and to
$-\infty$ as $p\uparrow1$.  Thus the maximizer exists in $(0,1)$ and is
unique.  Finally,
\begin{equation}
 F_t(p,u)=F_t(1-p,2t-u).
\end{equation}
Uniqueness gives the first identity in
\eqref{eq:numerical-profile-symmetry}, and the second follows by
substitution.  Since $[a,b]$ is symmetric about $t=(a+b)/2$, the outer
maximization may be restricted to $[a,t]$.
\end{IEEEproof}

For the representative channel below, Appendix~\ref{app:binary-analytic-maximization}
localizes the unique fixed-$u$ maximizer to a narrow strip on the left half
of the parameter interval.  Elementary bounds on that strip make the
optimized profile strictly concave.  Together with the symmetry in
\eqref{eq:numerical-profile-symmetry}, this identifies its unique maximum.

For the representative channel
\begin{equation}\label{eq: representative channel}
 a=\frac{1}{10},\qquad b=\frac{1}{5},\qquad t=\frac{3}{20},
\end{equation}
define
\begin{equation}
 M_{\rm N}:=\max_{p_1,p_2\in[0,1]}H(X_1\mid Z).
 \label{eq:numerical-MN-definition}
\end{equation}
We separate the exact-region identification into three steps.

Theorem~\ref{thm:binary-nonadaptive-converse} gives the operational converse
\begin{equation}
 \mathcal C_{\rm N}^{(1)}
 \subseteq[0,M_{\rm N}]\times[0,\log2].
 \label{eq:numerical-RN-MN-outer-bound}
\end{equation}
At this stage no value of $M_{\rm N}$ has been asserted.

\textit{Step 2: Analytic identification of $M_{\rm N}$.}
Appendix~\ref{app:binary-analytic-maximization} proves
\begin{equation}
 M_{\rm N}=h_2\!\left(\frac{3}{20}\right).
 \label{eq:numerical-MN-exact-definition}
\end{equation}
It first localizes the unique fixed-$u$ maximizer $p^\star(u)$ to a narrow
strip for $1/10\leq u\leq3/20$.  Direct rational bounds show that the
Hessian is negative definite on this strip.  The optimized profile is
therefore strictly concave, and symmetry identifies its unique maximizer as
$(p,u)=(1/2,3/20)$.

\textit{Step 3: Analytic reverse inclusion.}
The uniform product input places $(\log2,0)$ and $(0,\log2)$ in
$\mathcal R_{\rm A}$, while it places $(0,\log2)$ in
$\mathcal R_{\rm N}$.  Thus
\begin{align}
 \max_{(R_1,R_2)\in\mathcal R_{\rm A}}R_1&=\log2,&
 \max_{(R_1,R_2)\in\mathcal R_{\rm A}}R_2&=\log2,\notag\\
 \max_{(R_1,R_2)\in\mathcal R_{\rm N}}R_2&=\log2.
 \label{eq:numerical-overall-coordinate-maxima}
\end{align}
For the same uniform product input, the fixed-distribution non-adaptive
rectangle is
\begin{equation}
 \left[0,h_2\!\left(\frac{3}{20}\right)\right]
 \times[0,\log2]
 \subseteq\mathcal R_{\rm N}.
 \label{eq:numerical-uniform-rectangle}
\end{equation}
Combining this reverse inclusion with
\eqref{eq:numerical-RN-MN-outer-bound} and the analytic identity
\eqref{eq:numerical-MN-exact-definition} gives the complete non-adaptive capacity region
\begin{equation}
 \mathcal C_{\rm N}^{(1)}=\mathcal R_{\rm N}
 =\left[0,h_2\!\left(\frac{3}{20}\right)\right]
  \times[0,\log2].
 \label{eq:numerical-RN-exact-rectangle}
\end{equation}
Numerically, $h_2(3/20)\approx0.4227091$ nats,
$\log2\approx0.6931472$ nats, and the $R_1$-axis gap is approximately
$0.2704381$ nats.

For comparison, substituting the uniform product input into the adaptive
fixed-distribution formula gives
\begin{align}
 \mathcal R_{\rm A}(P_{\rm unif})
 =\Bigl\{(R_1,R_2)&\in\mathbb R_+^2:
 R_1\leq\log2,\ R_2\leq\log2,\notag\\
 &R_1+R_2\leq\log2+h_2(3/20)\Bigr\}.
 \label{eq:numerical-uniform-RA}
\end{align}
Its non-axis boundary is the exact line
\begin{equation}
 R_1+R_2=\log2+h_2\!\left(\frac{3}{20}\right),
 \label{eq:numerical-uniform-RA-sloping-boundary}
\end{equation}
joining $(\log2,h_2(3/20))$ and $(h_2(3/20),\log2)$.

Figure~\ref{fig:numerical-overall-region-comparison-ab-010-020} compares the
exact overall non-adaptive region with this fixed-uniform-input adaptive
region.  Since
$\mathcal R_{\rm A}(P_{\rm unif})\subseteq\mathcal R_{\rm A}$, it displays an
explicit adaptive subset extending strictly beyond the complete non-adaptive capacity region.
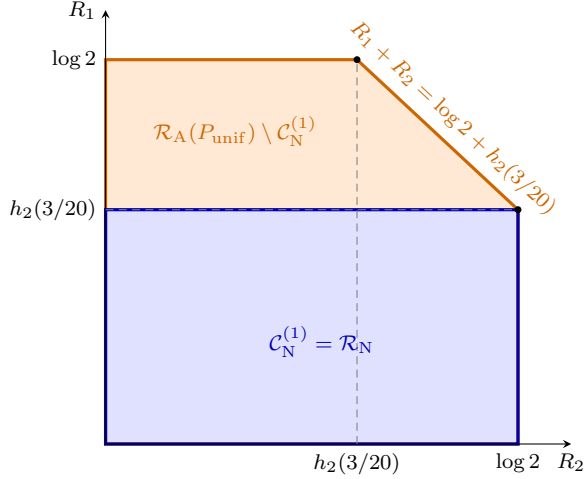
\begin{figure}[t]
 \centering
 \begin{tikzpicture}[x=0.88cm,y=0.82cm,>=stealth,font=\footnotesize]
  \def\L{6.20}
  \def\H{3.78}
  \coordinate (O) at (0,0);
  \coordinate (L0) at (\L,0);
  \coordinate (LH) at (\L,\H);
  \coordinate (HL) at (\H,\L);
  \coordinate (H0) at (\H,0);
  \coordinate (OH) at (0,\H);
  \coordinate (OL) at (0,\L);
  \fill[orange!17] (O)--(L0)--(LH)--(HL)--(OL)--cycle;
  \fill[blue!12] (O)--(L0)--(LH)--(OH)--cycle;
  \draw[orange!80!black,very thick]
    (O)--(L0)--(LH)--(HL)--(OL)--cycle;
  \draw[blue!70!black,very thick]
    (O)--(L0)--(LH)--(OH)--cycle;
  \draw[->] (0,0)--(7.0,0) node[below] {$R_2$};
  \draw[->] (0,0)--(0,7.0) node[left] {$R_1$};
  \draw[densely dashed,gray] (OH)--(LH);
  \draw[densely dashed,gray] (H0)--(HL);
  \node[left] at (OH) {$h_2(3/20)$};
  \node[left] at (OL) {$\log2$};
  \node[below] at (H0) {$h_2(3/20)$};
  \node[below] at (L0) {$\log2$};
  \node[blue!70!black,align=center] at (3.25,1.65)
    {$\mathcal C_{\rm N}^{(1)}=\mathcal R_{\rm N}$};
  \node[orange!80!black,align=center] at (1.95,5.05)
    {$\mathcal R_{\rm A}(P_{\rm unif})\setminus\mathcal C_{\rm N}^{(1)}$};
  \path (LH)--(HL)
    node[midway,sloped,above=3pt,orange!80!black,
         fill=white,inner sep=1.2pt]
    {$R_1+R_2=\log2+h_2(3/20)$};
  \fill (LH) circle (1.3pt);
  \fill (HL) circle (1.3pt);
 \end{tikzpicture}
 \caption{Exact overall non-adaptive region and an explicit fixed-uniform-input
 adaptive subregion for $a=0.1$ and $b=0.2$.  The blue rectangle is the complete non-adaptive capacity region
 $\mathcal C_{\rm N}^{(1)}=\mathcal R_{\rm N}
=[0,h_2(3/20)]\times[0,\log2]$, as proved in
 Appendix~\ref{app:binary-analytic-maximization}.  The orange polygon is
 the fixed-uniform-input adaptive region
 $\mathcal R_{\rm A}(P_{\rm unif})$, not an identification of the complete
 overall adaptive inner region.  Its sloping boundary is the exact line
 $R_1+R_2=\log2+h_2(3/20)$.  Here
 $h_2(3/20)\approx0.4227091$, $\log2\approx0.6931472$, and their difference
 is approximately $0.2704381$ nats.}
 \label{fig:numerical-overall-region-comparison-ab-010-020}
\end{figure}

\section{Conclusion}
\label{sec:conclusion}

We developed a key-assisted adaptive coding scheme for two-way wiretap
channels under strong one-sided secrecy.  The multiround proof separates the
one-round resolvability estimate from the propagation of a previously
generated key.  It explicitly retains prior key leakage, treats one-time-pad
use with side information, and transfers the ideal true-key analysis to the
implemented process through a coupling up to the first wrongly used key.
Simultaneous codebook selection then yields one codebook sequence satisfying
reliability and secrecy simultaneously, and a growing-round construction
removes the initialization-rate loss.

For the resulting constructions, exact elimination of the component rates
shows that the adaptive fixed-distribution region contains the
non-adaptive rectangle and gives a complete description of the User~1 gain
profile.  The best sum-rate is unchanged.  Separately, for an open binary
channel family, a direct converse for arbitrary non-adaptive stochastic
codes proves an operational adaptive-versus-non-adaptive separation.  For a
representative channel, the converse and an exact analytic maximization identify
the complete non-adaptive strong one-sided secrecy capacity region.

A supplementary appendix compares the direct one-sided component system
with the formal one-sided specialization of an earlier symmetric system.
That algebraic comparison is not needed for the multiround theorem or the
binary operational separation.  Broader capacity characterizations for
adaptive codes remain open.

\section*{Acknowledgements} 
During the preparation of this manuscript, the authors used Microsoft Copilot with the GPT-5.6 Thinking model to assist with language editing, organization and presentation of the manuscript, and the exploration, development, and checking of certain mathematical derivations and arguments. 
All AI-assisted material was critically reviewed, verified, and revised by the authors, who take full responsibility for the accuracy and integrity of the manuscript.

\appendices

\section{Formal comparison with the full symmetric component-rate system}
\label{sec:symmetric-relation}
The asymmetric key-assisted adaptive code construction was specially designed for one-sided secrecy.  For
comparison, this section extracts from the full symmetric key-exchange
construction of \cite{CH2025} only its component variables and one-sided
auxiliary-rate system.  We do not invoke the cited multiround leakage proof
for the implemented decoded-key process.  The comparison below is an
equality of projected formal component-rate regions, not a new operational
achievability theorem for the full construction.  

\subsection{Full symmetric component allocation}
In each payload round, the full system assigns to User~$i$ a protected payload of rate $R_{i,\mathrm{sec}}$, an encrypted payload of rate
$R_{i,e}$, a fresh key of rate $R_{i,k}$ for the other user, and an open
randomization coordinate of rate $R_{i,o}$.  Its payload and internal index rates
are
\begin{align}
R_i&=R_{i,\mathrm{sec}}+R_{i,e},
\label{eq:full-payload-rates}\\
\widetilde R_i&=R_{i,\mathrm{sec}}+R_{i,k},&
\widetilde R_{i,r}&=R_{i,e}+R_{i,o}.
\label{eq:full-underlying-rates}
\end{align}
The formal full system uses the original key-size conditions
\begin{equation}
R_{1,e}\leq R_{2,k},\qquad R_{2,e}\leq R_{1,k}.
\label{eq:key-matching}
\end{equation}
For payload-rate projection, any excess in $R_{2,k}-R_{1,e}$ may be removed
without violating feasibility.  The direct asymmetric system is therefore
represented in the equality-normalized form $R_{2,k}=R_{1,e}$, without
redefining the cited full system.
Note that in this cited symmetric construction, the first-round message components are fixed
during initialization. The asymmetric construction proposed in this paper instead uses
independent dummy coordinates to keep the codebook-index dimensions uniform.
This difference does not enter the component-rate projection below.

\subsection{One-sided auxiliary-rate system}
At the Shannon-information level, the formal one-sided specialization is
\allowdisplaybreaks
\begin{align}
\widetilde R_1+\widetilde R_{1,r}&<A,
\label{eq:full-rel1}\\
\widetilde R_2+\widetilde R_{2,r}&<B,
\label{eq:full-rel2}\\
\widetilde R_{1,r}&>C,
\label{eq:full-sec1}\\
\widetilde R_{2,r}+R_{2,\mathrm{sec}}&>D,
\label{eq:full-sec2}\\
\widetilde R_{1,r}+\widetilde R_{2,r}+R_{2,\mathrm{sec}}&>E.
\label{eq:full-sec12}
\end{align}
Together with \eqref{eq:full-payload-rates}--\eqref{eq:key-matching}
and nonnegativity, these inequalities define the strict full component-rate
system used here.  This algebraic specialization does not assert that the
cited analysis already provides the decoded-key transfer established for
the direct construction in Section~\ref{sec:embedded-key-propagation}.

\subsection{Equality of the projected component-rate regions}
\begin{proposition}[Equality of the projected component-rate regions]
\label{prop:equal-projections}
Fix $P\in\mathcal Q_{\rm F}$ and $(R_1,R_2)$ with $R_1>0$.  The formal full
component-rate system above has a strictly feasible assignment realizing
$(R_1,R_2)$ if and only if the asymmetric component-rate system has one.
Consequently, after taking closures, the two formal systems have the same
projected achievable rate region $\mathcal R_{\rm A}(P)$.
\end{proposition}
\begin{IEEEproof}
\textit{Asymmetric to full.}  Every asymmetric tuple is a full tuple after
setting
\begin{equation}
R_{1,k}=0,\qquad R_{2,e}=0,\qquad R_{2,\mathrm{sec}}=R_2.
\label{eq:asym-to-full-embedding}
\end{equation}
The remaining components are unchanged.  Both payloads, both reliability
sums, all three averaging sums, and User~1's key-matching condition are
preserved; the second key-matching condition becomes $0\leq0$.

\textit{Full to asymmetric.}  From a strictly feasible full tuple define
\begin{align}
&R_{1,\mathrm{sec}}'=R_{1,\mathrm{sec}}, \quad R_{1,e}' =R_{1,e}, \quad R_{1,o}'=R_{1,o}+R_{1,k},
\label{eq:full-to-asym-user1}\\
&R_2'=R_{2,\mathrm{sec}}+R_{2,e}=R_2, \quad R_{2,k}' =R_{1,e}, \quad R_{2,o}'=R_{2,o}.
\label{eq:full-to-asym-user2}
\end{align}
Then $R_1'=R_1$, $R_2'=R_2$, and $R_{1,e}'=R_{2,k}'$.  The transformed
User~1 reliability load is
\begin{align}
R_{1,\mathrm{sec}}'+R_{1,e}'+R_{1,o}'
&=\widetilde R_1+\widetilde R_{1,r}<A,
\label{eq:full-to-asym-rel1}
\end{align}
and the User~2 load satisfies
\begin{align}
R_2'+R_{2,k}'+R_{2,o}'
&\leq\widetilde R_2+\widetilde R_{2,r}<B.
\label{eq:full-to-asym-rel2}
\end{align}
The transformed averaging rates obey
\allowdisplaybreaks
\begin{align}
&R_{1,e}'+R_{1,o}'=\widetilde R_{1,r}+R_{1,k}>C,
\label{eq:full-to-asym-sec1}\\
&R_2'+R_{2,o}'=R_{2,\mathrm{sec}}+\widetilde R_{2,r}>D,
\label{eq:full-to-asym-sec2}\\
&(R_{1,e}'+R_{1,o}')+(R_2'+R_{2,o}')
\geq\widetilde R_{1,r}+\widetilde R_{2,r}+R_{2,\mathrm{sec}}>E.
\label{eq:full-to-asym-sec12}
\end{align}
Thus the transformations preserve the payload pair and strict feasibility.
The $R_1=0$ boundary and all remaining boundary points follow by the closure
argument of Section~\ref{sec:region-geometry}.
\end{IEEEproof}

\begin{remark}[Scope of the comparison]
This proves equality of the projected formal component-rate systems.  It
does not identify the implemented stochastic processes or supply a separate
full-construction decoded-key leakage theorem.
\end{remark}

\section{Proof of Proposition~\ref{prop:selected-index-component}}
\label{app:selected-index-component-proof}
Compose the physical eavesdropper channel with the two stochastic prefixes.
The resolvability argument is therefore applied to the induced memoryless
MAC $P_{Z\mid V_1V_2}$ in \eqref{eq:induced-eve-mac}; no additional
randomness from the prefixes remains outside this channel.

Let
\begin{equation}
 L_i:=(J_i,B_i),\qquad R_{L_i}:=R_{J_i}+R_{B_i},
 \label{eq:selected-superindices}
\end{equation}
and fix a retained pair $(A_1,A_2)=(a_1,a_2)$.  For a codebook realization
$c$, averaging the two super-indices gives $P^{c}_{Z^n\mid a_1,a_2}(z^n)$ by
\begin{align}
 \frac{1}{|\mathcal L_1||\mathcal L_2|}
 \sum_{l_1,l_2}
 P_{Z\mid V_1V_2}^{\otimes n}
 \bigl(z^n\mid v_1^n(a_1,l_1;c),v_2^n(a_2,l_2;c)\bigr).
 \label{eq:selected-fixed-retained-output}
\end{align}
Because all complete-index codewords are generated independently, the two
sub-codebooks exposed after fixing $(a_1,a_2)$ remain independent i.i.d.
codebooks with single-letter laws $P_{V_1}$ and $P_{V_2}$ and rates
$R_{L_1}$ and $R_{L_2}$.

Use the product reference output
\begin{align}
 Q_{Z^n}:=&P_Z^{\otimes n},\notag\\
 P_Z(z):=&\sum_{v_1,v_2}P_{V_1}(v_1)P_{V_2}(v_2)
 P_{Z\mid V_1V_2}(z\mid v_1,v_2).
 \label{eq:selected-reference-output}
\end{align}
Applying \cite[Lemma~2]{HC2023} to the conditional sub-codebook gives,
for every retained pair, the same ensemble bound
\begin{align}
 &\mathbb E_{\mathsf C}
 D\bigl(P^{\mathsf C}_{Z^n\mid a_1,a_2}\Vert Q_{Z^n}\bigr)\notag\\
 &\leq\frac1s\Bigl[
 e^{ns(E_s-R_{L_1}-R_{L_2})}
 +e^{ns(C_s-R_{L_1})}
 +e^{ns(D_s-R_{L_2})}\Bigr].
 \label{eq:selected-divergence-bound}
\end{align}
The three terms are the three nonempty transmitter subsets
$\{1,2\}$, $\{1\}$, and $\{2\}$, respectively.  Tensor-product additivity
produces the $n$-fold exponents, and relative entropy is bounded by the
order-$(1+s)$ R\'{e}nyi quantity used in the cited lemma.

For a fixed codebook $c$, let
\begin{equation}
 P^c_{Z^n}:=\frac{1}{|\mathcal A_1||\mathcal A_2|}
 \sum_{a_1,a_2}P^c_{Z^n\mid a_1,a_2}.
\end{equation}
The standard divergence decomposition is
\begin{align}
 &\frac{1}{|\mathcal A_1||\mathcal A_2|}
 \sum_{a_1,a_2}
 D\bigl(P^c_{Z^n\mid a_1,a_2}\Vert Q_{Z^n}\bigr)\notag\\
 &\qquad=I_c(A_1,A_2;Z^n)+D(P^c_{Z^n}\Vert Q_{Z^n}).
 \label{eq:selected-divergence-decomposition}
\end{align}
The second term is nonnegative.  Averaging
\eqref{eq:selected-divergence-decomposition} over the random codebook and
using \eqref{eq:selected-divergence-bound}, which is uniform in the retained
pair, yields \eqref{eq:selected-index-component-bound} after substituting
$R_{L_i}=R_{J_i}+R_{B_i}$.

For the history-conditioned statement, fix a positive-probability history
realization $h$.  By hypothesis, conditional on $h$ the six current
coordinates retain the same product-uniform law, and the current codebook is
independent of $h$ with its original product ensemble distribution.
Consequently, fixing the retained pair again exposes the same independent
sub-codebooks, while the induced MAC in \eqref{eq:induced-eve-mac}, the
reference output in \eqref{eq:selected-reference-output}, and the quantities
$C_s,D_s,E_s$ are unchanged.  Equations
\eqref{eq:selected-divergence-bound} and
\eqref{eq:selected-divergence-decomposition} therefore hold for every such
$h$ with the same right-hand side.  Averaging the fixed-history inequality
over the history and the codebooks outside the current round gives the
conditional version by the tower property.

\section{Proof of Proposition~\ref{prop:nonadaptive-one-block}}
\label{app:nonadaptive-one-block-proof}
User~2 decodes User~1's complete index pair, of total rate
$R_1+R_{1,r}$, through the induced channel $W_1$ in
\eqref{eq:induced-reliability-channels}; User~1 similarly decodes User~2's
complete pair, of rate $R_2+R_{2,r}$, through $W_2$.  The i.i.d.
random-coding maximum-likelihood bound, in its Gallager form as used in
\cite[Lemma~4]{HC2023} and with the convention
\eqref{eq:downward-renyi-definition}, gives the two exponential terms in
\eqref{eq:nonadaptive-error-ensemble-an}.  More precisely, the Gallager
functions $A_s(P)$ and $B_s(P)$ there are the two conditional downward
R\'enyi-information terms in \eqref{eq:nonadaptive-error-ensemble-an}, with
order $1/(1+s)$ and receiver side information averaged before the logarithm.  Their sum follows from the
union bound.

For secrecy, apply Proposition~\ref{prop:selected-index-component} with
\begin{align}
&A_1=M_1,\quad  J_1=\mathrm{const},\quad 
 B_1=\text{User~1's randomization index},\notag\\
&A_2=\mathrm{const},\quad  J_2=M_2,\quad 
 B_2=\text{User~2's randomization index}.\notag
\end{align}  This gives the bound defining $b_n$ in
\eqref{eq:nonadaptive-leakage-ensemble-bn}, and hence
\eqref{eq:nonadaptive-leakage-ensemble}.  Finally apply the normalized-sum
criterion to $P_e^n/a_n$ and the leakage divided by $b_n$.  Its expectation
is at most two, so one realization satisfies both inequalities in
\eqref{eq:nonadaptive-deterministic-consequence}.  The leading factor two
is introduced only at this selection step.

\section{Proof of Lemma \ref{lem:otp-leakage-absorption}}
\label{app:otp-leakage-absorption-proof}
Put $E_t=Z^{1:t-1}$, $W=W_{t-1}$, and  
$D=(Q_t,\boldsymbol{\mathsf C})$.  The Markov requirement for the ideal observation kernel in condition
(A3) gives the  
conditional data-processing inequality  
\begin{align}  
	I(U_t;Z_t\mid E_t,W,D)&\leq I(U_t;C_t\mid E_t,W,D)\notag\\  
	&\leq I(U_t;C_t,E_t\mid W,D). \label{eq:otp-conditional-dpi}  
\end{align}  
The joint-distribution requirement in condition (A1a) makes $U_t$ fresh and uniform conditional on $D$ and  
independent of $(K_{t-1},W,E_t)$.  Condition (A1b), together with the  
independence of the fresh variable $Q_t$, gives  
$K_{t-1}\perp W\mid D$ before $E_t$ is observed.  Proposition~\ref{prop:abstract-otp}, applied with these choices, therefore yields  
\begin{align}  
	I(U_t;C_t,E_t\mid W,D)
	\leq I(K_{t-1};E_t\mid W,D).
\end{align}  
Finally, $Q_t$ is independent of  
$(K_{t-1},W,E_t)$ conditional on the codebooks, so the last  
term equals  
$I(K_{t-1};E_t\mid W,\boldsymbol{\mathsf C})$.  Combining  
these inequalities proves \eqref{eq:otp-leakage-absorption}.  No  
posterior independence of the preceding key and payload given $E_t$ is  
used.

\section{Proof of Lemma \ref{lem:one-step-recursion}}
\label{app:one-step-recursion-proof}
Write $A_t=W_{t-1}$ and $E_t=Z^{1:t-1}$.  Since  
$(W_t,K_t)$ is, up to ordering, $(A_t,Q_t,U_t)$, with the  
terminal key occurring only inside $Q_t$, the chain rule first gives  
\allowdisplaybreaks
\begin{align}  
	&I(W_t,K_t;Z^{1:t}\mid\boldsymbol{\mathsf C})\notag\\  
	&=I(A_t,Q_t,U_t;E_t\mid\boldsymbol{\mathsf C})
	+I(A_t,Q_t,U_t;Z_t\mid E_t,\boldsymbol{\mathsf C}).  
	\label{eq:recursion-split-past-current}  
\end{align}  
The independence condition (A1a) for the ideal process implies that $(Q_t,U_t)$ is independent of $(A_t,E_t)$ conditional on the  
codebooks.  Hence the past-observation term is exactly  
\begin{equation}  
	I(A_t,Q_t,U_t;E_t\mid\boldsymbol{\mathsf C})  
	=I(A_t;E_t\mid\boldsymbol{\mathsf C}).  
	\label{eq:recursion-past-term}  
\end{equation}  
Expand the current-output term in the order $Q_t,A_t,U_t$:  
\begin{align}  
	&I(A_t,Q_t,U_t;Z_t\mid E_t,\boldsymbol{\mathsf C})\notag\\  
	&=I(Q_t;Z_t\mid E_t,\boldsymbol{\mathsf C})
	+I(A_t;Z_t\mid E_t,Q_t,\boldsymbol{\mathsf C})\notag\\  
	&\quad+I(U_t;Z_t\mid E_t,A_t,Q_t,\boldsymbol{\mathsf C}).  
	\label{eq:recursion-current-expansion}  
\end{align}  
By (A1a), $Q_t\perp(A_t,K_{t-1},E_t)\mid\boldsymbol{\mathsf C}$.  The chain rule therefore gives
\begin{equation}
 I(Q_t;Z_t\mid E_t,\boldsymbol{\mathsf C})
 \leq I(Q_t;Z_t\mid E_t,A_t,K_{t-1},\boldsymbol{\mathsf C}).
 \label{eq:recursion-Q-bound}
\end{equation}
The second identity in (A3) gives
$I(A_t;Z_t\mid E_t,Q_t,\boldsymbol{\mathsf C})=0$.
Lemma~\ref{lem:otp-leakage-absorption} bounds the third term by  
\begin{equation}  
	I(K_{t-1};E_t\mid A_t,\boldsymbol{\mathsf C}).  
	\label{eq:recursion-U-bound}  
\end{equation}  
Substituting \eqref{eq:recursion-past-term}--  
\eqref{eq:recursion-U-bound} into  
\eqref{eq:recursion-split-past-current} yields  
\begin{align}  
	&I(W_t,K_t;Z^{1:t}\mid\boldsymbol{\mathsf C})\notag\\  
	&\leq I(A_t;E_t\mid\boldsymbol{\mathsf C})  
	+I(K_{t-1};E_t\mid A_t,\boldsymbol{\mathsf C})\notag\\  
	&\quad+I(Q_t;Z_t\mid A_t,K_{t-1},E_t,  
	\boldsymbol{\mathsf C})\notag\\  
	&=I(A_t,K_{t-1};E_t\mid\boldsymbol{\mathsf C})\notag\\  
	&\quad+I(Q_t;Z_t\mid A_t,K_{t-1},E_t,  
	\boldsymbol{\mathsf C}),  
\end{align}  
where the last equality is the chain rule.  Restoring  
$A_t=W_{t-1}$ and $E_t=Z^{1:t-1}$ proves  
\eqref{eq:one-step-recursion-revised}.  

\section{Auxiliary proofs for embedded-key leakage propagation}
\label{app:embedded-key-proofs}

\subsection{Proof of Proposition~\ref{prop:abstract-otp}}
Condition on $(W,D)=(w,d)$.  Since
$I(U;E\mid W=w,D=d)=0$, the chain rule gives the first equality below.
For fixed $U=u$, the group translation $C=u\oplus K$ preserves conditional
entropy, and hence
$H(C\mid U,E,W=w,D=d)=H(K\mid E,W=w,D=d)$.  Therefore
\begin{align}
&I(U;C,E\mid W=w,D=d)\notag\\
=&I(U;C\mid E,W=w,D=d)\notag\\
=&H(C\mid E,W=w,D=d)-H(K\mid E,W=w,D=d)\notag\\
\leq&\log|G|-H(K\mid E,W=w,D=d)\notag\\
=&I(K;E\mid W=w,D=d),
\end{align}
where the last equality uses that $K$ is uniform and independent of $W$
conditional on $D=d$, so $H(K\mid W=w,D=d)=\log|G|$.
Average over $(w,d)$.  The key statement follows by data processing.

\subsection{Proof of Lemma~\ref{lem:common-marginal-continuity}}
Because the two laws have the same $W$ marginal,
$I_P(W;Z)-I_Q(W;Z)=H_Q(W\mid Z)-H_P(W\mid Z)$.
For $d\leq1/2$, the classical conditional-entropy continuity bound
\cite{Winter2016} gives
\begin{equation}
|H_P(W\mid Z)-H_Q(W\mid Z)|
\leq2d\log|\mathcal W|+2h_2(d).
\end{equation}
Only the alphabet of the protected variable $W$ enters this bound; the
observation alphabet may be arbitrarily large.  This point is useful here
because $Z^{1:T}$ grows with the number of rounds.
For $d>1/2$, the left-hand side is at most $\log|\mathcal W|$, whereas
$2d\log|\mathcal W|+2h_2(1/2)\geq\log|\mathcal W|$.
Combining the two cases proves
\eqref{eq:common-marginal-continuity} for every $d\in[0,1]$.

\subsection{Proof of Proposition~\ref{prop:abstract-transfer}}
Fix a codebook collection $\boldsymbol c$.  The coupling makes the implemented
and ideal pairs $(W_T,Z^{1:T})$ identical unless an incorrect key is used.
Therefore
\begin{equation}
d(\boldsymbol c):=d_{\rm TV}
(P^{\boldsymbol c}_{\rm real},P^{\boldsymbol c}_{\rm ideal})
\leq p_{n,T}(\boldsymbol c).
\end{equation}
For each fixed $\boldsymbol c$, the two laws have the same $W_T$ marginal,
because the coupling changes only the codeword lookup and leaves the payload
and generated-key variables unchanged.  Lemma
\ref{lem:common-marginal-continuity} yields the fixed-codebook inequality
\begin{align}
I_{\rm real}(W_T;Z^{1:T}\mid\boldsymbol c)\leq& I_{\rm ideal}(W_T;Z^{1:T}\mid\boldsymbol c)\notag\\
&
+2d(\boldsymbol c)\log|\mathcal W_T|
+2h_2(\min\{d(\boldsymbol c),1/2\}).
\end{align}
The function $x\mapsto h_2(\min\{x,1/2\})$ is nondecreasing and concave on
$[0,1]$: it agrees with the increasing concave binary entropy on
$[0,1/2]$ and is constant thereafter.  Average the fixed-codebook inequality over the full codebook collection.
Then use $d(\boldsymbol c)\leq p_{n,T}(\boldsymbol c)$, monotonicity of the
correction term, and Jensen's inequality for its concave entropy part.
This gives \eqref{eq:abstract-transfer}--
\eqref{eq:abstract-transfer-eta}.

\subsection{Proof of Proposition~\ref{prop:abstract-selection}}
Apply the normalized-sum criterion
\[
X(\boldsymbol{\mathsf C})=
\frac{P_{e,\mathrm{actual}}(\boldsymbol{\mathsf C})}{a_{n,T}}+
\frac{I_{\mathrm{actual}}(W_T;Z^{1:T}\mid
\boldsymbol{\mathsf C})}{b_{n,T}}.
\]
Its expectation is at most two, so some realization satisfies $X\leq2$.
If $a_{n,T}=0$ or $b_{n,T}=0$, the corresponding nonnegative random
quantity has expectation zero and therefore vanishes almost surely; that
term is omitted from the normalized sum.

\section{Proof of Proposition~\ref{prop:exact-adaptive-projection}}
\label{app:exact-adaptive-projection-proof}
Introduce the aggregate variables
\begin{equation}
	x:=R_{1,e}+R_{1,o},\qquad
	y:=R_2+R_{2,o},\qquad e:=R_{1,e}.
	\label{eq:projection-aggregates}
\end{equation}
For fixed $(R_1,R_2,x,y,e)$, increasing $R_{2,k}$ only tightens
\eqref{eq:projection-rel2}.  Hence feasibility can be tested with the
smallest permitted value $R_{2,k}=e$.  The full component-rate system is
therefore equivalent to
\begin{align}
	x&>C, & y&>D, & x+y&>E,
	\label{eq:projection-xyz-secrecy}\\
	y&\geq R_2,
	& R_1+x-e&<A,
	& e+y&<B,
	\label{eq:projection-xyz-reliability}\\
	0&\leq e\leq\min\{R_1,x\}.
	\label{eq:projection-e-range}
\end{align}
The equivalence is constructive: given $(x,y,e)$, take
\begin{align}
	R_{1,e}&=e,&R_{1,o}&=x-e,&R_{1,\mathrm{sec}}&=R_1-e,\\
	R_{2,k}&=e,&R_{2,o}&=y-R_2.
\end{align}

We eliminate $e$ first.  For fixed $x$ and $y$, a strictly feasible $e$
exists if and only if
\begin{equation}
	\max\{0,R_1+x-A\}
	<\min\{R_1,x,B-y\}.
	\label{eq:projection-e-interval}
\end{equation}
Combining this interval condition with
\eqref{eq:projection-xyz-secrecy} and $y\geq R_2$, and then eliminating
$x$ and $y$, gives
\begin{align}
	R_1&<A,\qquad R_2<B,\notag\\
	R_1&<A+B-\max\{E,C+D\},
	\label{eq:projection-preliminary}\\
	R_1+R_2&<A+B-C.
	\notag
\end{align}
For completeness, the implications can also be read directly from
\eqref{eq:projection-e-interval}: the non-strict upper bounds $e\leq R_1$ and $e\leq x$, together with
the strict lower bound in \eqref{eq:projection-e-interval}, produce
$R_1<A$; $e<B-y$ together with $y\geq R_2$ produces $R_2<B$;
the simultaneous lower bounds $x>C$, $y>D$, and $x+y>E$ produce the third
and fourth projected constraints.  Sufficiency is supplied by the constructive open-polytope elimination in
Appendix~\ref{app:projection-details}: it selects $x,y$ with $C<x<A$, $\max\{D,R_2\}<y<B$, and $E<x+y<A+B-R_1$, then chooses $e$ from \eqref{eq:projection-e-interval} and reconstructs all component rates.

Because $V_1$ and $V_2$ are independent under every $P\in\mathcal Q$,
\begin{equation}
	E=C+D+I(V_1;V_2\mid Z)\geq C+D.
	\label{eq:E-dominates-CD-journal}
\end{equation}
Thus the third bound in \eqref{eq:projection-preliminary} reduces to
$R_1<A+B-E$.  This proves necessity and sufficiency of
\eqref{eq:adaptive-strict-A}--\eqref{eq:adaptive-strict-C} for $R_1>0$.  For $R_1=0$, use positive-$R_1$ interior pairs converging to the target.
When the target also has $R_2=B$, both coordinates are backed off, for
example $R_1^{(k)}\downarrow0$ and $R_2^{(k)}\uparrow B$ with
$R_2^{(k)}<B$; this preserves every strict inequality and avoids a
degenerate strict $e$-interval.  Taking the
closure gives \eqref{eq:adaptive-fixed-P-region}.

\section{Additional details for proof of Theorem \ref{thm:adaptive-achievable-region}: Step 1}
\label{app:adaptive-achievable-region-proof-step1}
The constructive part of Proposition~\ref{prop:exact-adaptive-projection} selects aggregate rates
$(x,y,e)$ satisfying \eqref{eq:projection-xyz-secrecy}--
\eqref{eq:projection-e-range} strictly and then reconstructs
\begin{align}
	R_{1,e}&=e,&R_{1,o}&=x-e,&R_{1,\mathrm{sec}}&=R_1-e,\notag\\
	R_{2,k}&=e,&R_{2,o}&=y-R_2.
	\label{eq:adaptive-achievability-reconstruction}
\end{align}
All component rates are nonnegative, the key-size condition holds with
$R_{2,k}=R_{1,e}$, and
\eqref{eq:projection-rel1}--\eqref{eq:projection-sec12} hold with a
common positive Shannon-information slack.  This is the component system
used by the probabilistic construction, not merely a necessary projection.

For a point on the boundary $R_1=0$, take positive-$R_1$ interior
points converging to the target while keeping $R_2<B$; if the target also
has $R_2=B$, back off both coordinates before taking the limit.  This uses
the same strict open component system throughout.  Thus the positive-$R_1$
restriction in the exact strict projection removes no point from the closed
fixed-$P$ region.

\section{Proof of Theorem \ref{thm:fixed-P-gain}}
\label{app:fixed-P-gain-proof}
The inequality $r_{\rm N}^{\max}\leq r_{\rm A}^{\max}$ gives the individual-rate inclusion.  Moreover,
every point in $\mathcal R_{\rm N}(P)$ satisfies
\begin{equation}
	R_1+R_2\leq r_{\rm N}^{\max}+B\leq A+B-C=s_{\max},
\end{equation}
and hence satisfies all inequalities defining $\mathcal R_{\rm A}(P)$.
Finally,
\begin{equation}
	r_{\rm A}^{\max}-r_{\rm N}^{\max}=\max\{C,E-B\}-\max\{0,E-B\}.
\end{equation}
This difference is positive exactly when $C>0$ and $E-B<C$, which is
\eqref{eq:strictness-condition}.  If the difference is zero, the two regions
coincide; if it is positive, the adaptive region contains points with
$R_1>r_{\rm N}^{\max}$, so the inclusion is strict.

\section{Pairwise elimination for the asymmetric projection}
\label{app:projection-details}

We give the algebra underlying Proposition~\ref{prop:exact-adaptive-projection}.
Starting from \eqref{eq:projection-xyz-secrecy}--\eqref{eq:projection-e-range},
strict feasibility of $e$ is equivalent to
\begin{equation}
\max\{0,R_1+x-A\}<\min\{R_1,x,B-y\}.
\end{equation}
Pairing every lower bound with every upper bound yields
\begin{align}
R_1&>0,\qquad x>0,\qquad y<B,\notag\\
R_1&<A,\qquad x<A,
\label{eq:app-e-pairs}\\
R_1+x+y&<A+B.
\notag
\end{align}
Together with $x>C$, $y>\max\{D,R_2\}$, and $x+y>E$, the variables $x,y$
must therefore satisfy
\begin{align}
C<x&<A,\\
\max\{D,R_2\}<y&<B,\\
E<x+y&<A+B-R_1.
\end{align}
Let $L=\max\{D,R_2\}$ and $U=A+B-R_1$.  The attainable sums of pairs with $C<x<A$ and $L<y<B$ form the open interval $(C+L,A+B)$.  Hence the additional requirement $E<x+y<U$ is feasible exactly when
\begin{equation}
\max\{E,C+L\}<U.
\end{equation}
To construct a pair, choose $z$ strictly between these endpoints, choose $x$ in the nonempty intersection $(C,A)\cap(z-B,z-L)$, and set $y=z-x$.  Expanding these interval conditions gives
\begin{align}
C&<A,\\
D&<B,\\
R_2&<B,\\
E&<A+B,\\
R_1&<A,\\
R_1&<A+B-E,\\
R_1+R_2&<A+B-C,\\
R_1&<A+B-C-D.
\end{align}
The first, second, and fourth inequalities are precisely the strict
feasibility conditions defining $\mathcal Q_{\rm F}$.  Independence of
$V_1$ and $V_2$ gives $E\geq C+D$, so
$A+B-E\leq A+B-C-D$ and the last inequality is redundant.  The remaining
payload inequalities are exactly
\eqref{eq:adaptive-strict-A}--\eqref{eq:adaptive-strict-C}.

This also completes the projection analysis for the full component system. Indeed, Proposition~\ref{prop:equal-projections} gives a constructive equivalence between the full and asymmetric component systems and identifies their common closed fixed-distribution achievable rate region, so no second elimination is required.

\section{Analytic maximization for the representative binary channel}
\label{app:binary-analytic-maximization}

This appendix proves the exact maximization used in
Section~\ref{subsec:binary-numerical-illustration} without a global
concavity claim for the full two-dimensional domain.  We use the
one-dimensional reduction of Proposition~\ref{prop:numerical-profile-reduction},
localize its stationary maximizer to a narrow strip, and establish strict
concavity of the optimized profile on the left half of the parameter
interval.

Recall \eqref{eq:numerical-profile-F} and write $F:=F_{t}:$
\begin{equation}
	F(p,u)=h_2(p)+(1-p)h_2(u)+p h_2\!\left(u+d\right)
	-h_2\!\left(u+dp\right),
	\label{eq:analytic-representative-F}
\end{equation}
where $d=1-a-b$ and $a\leq u\leq b.$ Moreover, put 
\begin{equation}
	r:=u+dp,\quad
	L(x):=\log\frac{1-x}{x},\quad
	J(x):=\frac{1}{x(1-x)}.
	\label{eq:analytic-rLJ}
\end{equation}
Direct differentiation gives
\begin{align}
	F_p&=L(p)+h_2\!\left(u+d\right)-h_2(u)
	-dL(r),
	\label{eq:analytic-Fp}\\
	F_u&=(1-p)L(u)+pL\!\left(u+d\right)-L(r),
	\label{eq:analytic-Fu}\\
	F_{pp}&=-J(p)+d^2J(r),
	\label{eq:analytic-Fpp}\\
	F_{uu}&=J(r)-(1-p)J(u)-pJ\!\left(u+d\right),
	\label{eq:analytic-Fuu}\\
	F_{pu}&=dJ(r)-
	\left[L(u)-L\!\left(u+d\right)\right].
	\label{eq:analytic-Fpu}
\end{align}

The analysis below has a focus on the representative channel as defined in \eqref{eq: representative channel}, i.e., $a=1/10$ and $b=1/5$ and accordingly
\begin{equation}
 t=\frac{3}{20},\quad d=\frac{7}{10},\quad r=u+\frac{7}{10}p.
 \label{eq:analytic-representative-td}
\end{equation}
Proposition~\ref{prop:numerical-profile-reduction} already shows that, for
each fixed $u$, $F(\cdot,u)$ has a unique interior maximizer, denoted by
$p^\star(u)$.

\begin{lemma}[Localization of the stationary maximizer]
\label{lem:analytic-localization}
For $1/10\leq u\leq3/20$,
\begin{equation}
 \frac12\leq p^\star(u)<\frac{13}{25}.
 \label{eq:analytic-localization}
\end{equation}
\end{lemma}
\begin{IEEEproof}
Consider the strip
\begin{equation}
 \mathcal S:=\left[\frac12,\frac{13}{25}\right]
 \times\left[\frac1{10},\frac3{20}\right].
 \label{eq:analytic-strip}
\end{equation}
On $\mathcal S$, we have
\begin{equation}
 \frac9{20}\leq r\leq\frac{257}{500},
 \qquad 4\leq J(r)\leq\frac{400}{99}.
 \label{eq:analytic-rJ-bounds}
\end{equation}
Hence the positive term $dJ(r)$ in \eqref{eq:analytic-Fpu} satisfies
\begin{equation}
 \frac7{10}J(r)\leq\frac{280}{99}.
 \label{eq:analytic-dJ-upper}
\end{equation}
We next lower-bound the term $L(u)-L(u+d)$ that is subtracted from \eqref{eq:analytic-Fpu}.  The function
$u\mapsto L(u)-L(u+d)$ is decreasing for $d=7/10$ on $[1/10,3/20]$, and therefore
\begin{equation}
 L(u)-L\!\left(u+\frac7{10}\right)
 \geq2\log\frac{17}{3}.
 \label{eq:analytic-L-difference-lower}
\end{equation}
Consequently, the two terms in \eqref{eq:analytic-Fpu} satisfy the direct
comparison
\begin{equation}
 F_{pu}
 =\frac7{10}J(r)-\left[L(u)-L\!\left(u+\frac7{10}\right)\right]
 \leq\frac{280}{99}-2\log\frac{17}{3}<0
 \label{eq:analytic-Fpu-negative}
\end{equation}
throughout $\mathcal S$.
Note that the last inequality can be shown as follows. 
To calculate $\log\frac{17}{3},$ we use
\begin{equation}
	\log\frac{1+y}{1-y}
	=2\sum_{k=0}^{\infty}\frac{y^{2k+1}}{2k+1},
	\qquad 0<y<1.
	\label{eq:analytic-atanh-series}
\end{equation}
with $y=7/10.$ This gives
\begin{equation}
	2\log\frac{17}{3}
	>4\left(\frac7{10}+\frac{(7/10)^3}{3}\right)
	=\frac{2443}{750}>\frac{280}{99}.
	\label{eq:analytic-cross-strict}
\end{equation}

Since $F_p(1/2,3/20)=0$, it follows that for $\frac1{10}\leq u\leq\frac3{20},$
\begin{equation}
 F_p(1/2,u)\geq0.
 \label{eq:analytic-left-Fp}
\end{equation}

At the other vertical boundary, expansion of the binary entropies in
\eqref{eq:analytic-Fp} gives
\begin{align}
10F_p\!\left(\frac{13}{25},\frac1{10}\right)
&=\log\frac{2\cdot3^{28}\cdot29^7}{13^{10}\cdot67^7}.
 \label{eq:analytic-endpoint-log}
\end{align}
The exact integer comparison
\begin{align}
13^{10}67^7-2\cdot3^{28}29^7
=46276881052077093674329>0
 \label{eq:analytic-endpoint-integer}
\end{align}
shows that
$F_p(13/25,1/10)<0$.  Since $F_{pu}<0$ on $\mathcal S$, we have 
\begin{equation}
 F_p(13/25,u)<0,
 \qquad \frac1{10}\leq u\leq\frac3{20}.
 \label{eq:analytic-right-Fp}
\end{equation}
Finally, $F_{pp}<0$ by
\eqref{eq:numerical-variance-identity}.  Hence the unique zero of
$p\mapsto F_p(p,u)$ lies between the two vertical boundaries, proving
\eqref{eq:analytic-localization}.
\end{IEEEproof}

\begin{lemma}[Negative definiteness on the localized strip]
\label{lem:analytic-strip-hessian}
The Hessian of $F$ is negative definite throughout $\mathcal S$.
\end{lemma}
\begin{IEEEproof}
First, \eqref{eq:analytic-rJ-bounds} and $J(p)\geq4$ give
\begin{equation}
 F_{pp}\leq-4+\frac{49}{100}\frac{400}{99}
 =-\frac{200}{99}<0.
 \label{eq:analytic-Fpp-bound}
\end{equation}
Next, we evaluate $F_{uu}.$ Note that on $\mathcal S$, we have
\begin{equation}
 J(u)\geq\frac{400}{51},\qquad
 J\!\left(u+\frac7{10}\right)\geq\frac{25}{4}.
 \label{eq:analytic-J-side-bounds}
\end{equation}
Because $J(u)>J(u+7/10)$ on the relevant intervals and
$p\leq13/25$,
\begin{align}
(1-p)J(u)+pJ\!\left(u+\frac7{10}\right)
&\geq\frac{12}{25}\frac{400}{51}
 +\frac{13}{25}\frac{25}{4}
 =\frac{477}{68}.
 \label{eq:analytic-weighted-J-bound}
\end{align}
Together with $J(r)\leq\frac{400}{99}$ as given in \eqref{eq:analytic-rJ-bounds}, we obtain
\begin{equation}
 F_{uu}\leq\frac{400}{99}-\frac{477}{68}<0.
 \label{eq:analytic-Fuu-bound}
\end{equation}

It remains to control the mixed derivative. Recall that the function
$u\mapsto L(u)-L(u+d)$ (that is used in \eqref{eq:analytic-L-difference-lower}) is decreasing for $d=7/10$ on $[1/10,3/20],$ and thus satisfies  
\begin{equation}
 L(u)-L\!\left(u+\frac7{10}\right)\leq\log36.
 \label{eq:analytic-L-difference-upper}
\end{equation}
We already know that $F_{pu}<0$.  Since $J(r)\geq4$  as given in \eqref{eq:analytic-rJ-bounds} and
$\log36<4$,
\begin{equation}
 |F_{pu}|=
 L(u)-L\!\left(u+\frac7{10}\right)-\frac7{10}J(r)
 <4-\frac{14}{5}=\frac65.
 \label{eq:analytic-Fpu-bound}
\end{equation}
Here $\log36<4$ follows, for example, from
$e>1+1+1/2+1/6=8/3$ and hence $e^4>36$.
Combining \eqref{eq:analytic-Fpp-bound},
\eqref{eq:analytic-Fuu-bound}, and
\eqref{eq:analytic-Fpu-bound}, we obtain
\begin{align}
 \det\nabla^2F
 &=F_{pp}F_{uu}-F_{pu}^2\notag\\
 &>\frac{200}{99}
 \left(\frac{477}{68}-\frac{400}{99}\right)-\frac{36}{25}\notag\\
 &=\frac{19030538}{4165425}>0.
 \label{eq:analytic-determinant-bound}
\end{align}
Together with $F_{pp}<0$, this proves negative definiteness.
\end{IEEEproof}

\begin{proposition}[Exact analytic maximum for the representative channel]
\label{prop:analytic-representative-maximum}
For $a=1/10$ and $b=1/5$,
\begin{equation}
 \max_{\substack{0\leq p\leq1\\1/10\leq u\leq1/5}}F(p,u)
 =F\!\left(\frac12,\frac3{20}\right)
 =h_2\!\left(\frac3{20}\right).
 \label{eq:analytic-exact-maximum}
\end{equation}
The maximizer is unique.  Consequently,
\begin{equation}
 M_{\rm N}=h_2\!\left(\frac3{20}\right).
 \label{eq:analytic-MN-equality}
\end{equation}
\end{proposition}
\begin{IEEEproof}
Define the optimized profile
\begin{equation}
 G(u):=F(p^\star(u),u).
 \label{eq:analytic-profile}
\end{equation}
The strict inequality $F_{pp}<0$ and the implicit-function theorem show
that $p^\star(u)$ is continuously differentiable.  The envelope identity
and differentiation of $F_p(p^\star(u),u)=0$ give
\begin{align}
 G'(u)&=F_u(p^\star(u),u),\notag\\
 G''(u)&=F_{uu}-\frac{F_{pu}^2}{F_{pp}}
 =\frac{\det\nabla^2F}{F_{pp}},
 \label{eq:analytic-profile-second}
\end{align}
where the derivatives on the second line are evaluated at
$(p^\star(u),u)$.  By Lemma~\ref{lem:analytic-localization}, this point lies
in $\mathcal S$ for $1/10\leq u\leq3/20$.  Lemma~\ref{lem:analytic-strip-hessian}
therefore gives
\begin{equation}
 G''(u)<0,
 \qquad \frac1{10}<u<\frac3{20}.
 \label{eq:analytic-profile-concavity}
\end{equation}
The symmetry \eqref{eq:numerical-profile-symmetry} gives
$G(3/10-u)=G(u)$ and hence $G'(3/20)=0$.  Strict concavity on the left half
implies $G'(u)>0$ for $1/10\leq u<3/20$.  By symmetry, $G$ is strictly
decreasing on $3/20<u\leq1/5$.  Thus $u=3/20$ is its unique maximizer, and
symmetry together with uniqueness of the fixed-$u$ maximizer gives
$p^\star(3/20)=1/2$.  Substitution in
\eqref{eq:analytic-representative-F} proves
\eqref{eq:analytic-exact-maximum}.  The definition of $M_{\rm N}$ in
\eqref{eq:numerical-MN-definition} then gives
\eqref{eq:analytic-MN-equality}.
\end{IEEEproof}

\end{document}